\documentclass[11pt]{article}

\usepackage[bookmarks,colorlinks,breaklinks]{hyperref}  
\hypersetup{linkcolor=blue,citecolor=blue,filecolor=blue,urlcolor=blue} 
\usepackage{fullpage,appendix}
\usepackage{amsmath,amsfonts,amsthm,amssymb,xspace,bm}
\usepackage{dsfont}

\usepackage{tikz}
\newcommand{\newref}[2][]{\hyperref[#2]{#1~\ref*{#2}}}
\renewcommand{\eqref}[1]{\hyperref[#1]{(\ref*{#1})}}
\numberwithin{equation}{section}

\newcommand{\sref}[1]{\newref[Section]{#1}}

\newcommand{\tref}[1]{\newref[Theorem]{#1}}
\newcommand{\lref}[1]{\newref[Lemma]{#1}}

\newcommand{\cref}[1]{\newref[Corollary]{#1}}

\newcommand{\eref}[1]{\newref[Equation]{#1}}

\theoremstyle{plain}
\newtheorem{theorem}{Theorem}[section]
\newtheorem{lemma}[theorem]{Lemma}
\newtheorem{claim}[theorem]{Claim}

\newtheorem{corollary}[theorem]{Corollary}

\newtheorem{definition}[theorem]{Definition}
\newtheorem{fact}[theorem]{Fact}
\theoremstyle{definition}

\newcommand{\eat}[1]{}

\newcommand{\calU}{\ensuremath{\mathcal{U}}}
\newcommand{\E}{\mathop{\mathbb{E}\/}}

\newcommand{\poly}{\ensuremath{\mathrm{poly}}}

\renewcommand{\epsilon}{\ensuremath{\varepsilon}}
\newcommand{\zo}{\{0,1\}}

\newcommand{\lfta}{\ensuremath{\leftarrow}}
\newcommand{\R}{\mathbb R}

\newcommand{\pr}{\Pr}

\newcommand{\ignore}[1]{{}}
\newcommand{\bkets}[1]{\left(#1\right)}
\newcommand{\sbkets}[1]{\left[#1\right]}

\newcommand{\abs}[1]{\left|#1\right|}
\newcommand{\calD}{\mathcal{D}}
\newcommand{\calP}{\mathcal{P}}

\newcommand{\mnote}[1]{}
\newcommand{\pnote}[1]{}
\newcommand{\dnote}[1]{}

\title{Explicit resilient functions matching Ajtai-Linial}
\author{Raghu Meka\\Department of Computer Science\\University of California, Los Angeles}
\date{}

\begin{document}
\maketitle
\abstract{
A Boolean function on $n$ variables is $q$-resilient if for any subset of at most $q$ variables, the function is very likely to be determined by a uniformly random assignment to the remaining $n-q$ variables; in other words, no coalition of at most $q$ variables has significant influence on the function. Resilient functions have been extensively studied with a variety of applications in cryptography, distributed computing, and pseudorandomness. The best known resilient function on $n$ variables due to Ajtai and Linial \cite{AjtaiL93} has the property that only sets of size $\Omega(n/(\log^2 n))$ can have influence bounded away from zero. However, the construction of Ajtai and Linial is by the probabilistic method and does not give an efficiently computable function. 

We construct an explicit monotone depth three almost-balanced Boolean function on $n$ bits that is $\Omega(n/(\log ^2 n))$-resilient matching the bounds of Ajtai and Linial. The best previous explicit constructions of Meka \cite{Meka09} (which only gives a logarithmic depth function), and Chattopadhyay and Zuckerman \cite{ChattopadhyayZ15} were only $(n^{1-\beta})$-resilient for any constant  $0 < \beta < 1$. Our construction and analysis are motivated by (and simplifies parts of) the recent breakthrough of \cite{ChattopadhyayZ15} giving explicit two-sources extractors for polylogarithmic min-entropy; a key ingredient in their result was the construction of explicit constant-depth resilient functions.

An important ingredient in our construction is a new randomness-optimal oblivious sampler that  preserves moment generating functions of sums of variables and could be useful elsewhere. 
}

\section{Introduction}
In this work we study \emph{resilient functions} introduced by Ben-Or and Linial \cite{BenOrL85} in the context of collective-coin flipping. Consider the following game: There are $n$ players who communicate by broadcast and want to agree on a random coin-toss. If all the players are honest, this is trivial: pick a player, have the player toss a coin and use the resulting value as the collective coin-toss. Now suppose that there are a few \emph{bad} players who are computationally unbounded, can collude amongst themselves, and broadcast last in each round, i.e., they broadcast after observing the bits broadcast by the good players in each round. The problem of \emph{collective coin-flipping} is to design protocols so that the bad players cannot bias the collective-coin too much. An important and well-studied case of protocols are \emph{one-round} collective coin-flipping protocols. We will adopt the notation of boolean functions instead of protocols, as both are equivalent for a single round. 

\begin{definition}
For a Boolean function $f:\zo^n \to \zo$, and $Q \subseteq [n]$, let $I_Q(f)$ be the probability that $f$ is not-determined by a uniformly random partial assignment to the bits not in $Q$. Let $I_q(f) = \min_{Q\subseteq [n], |Q| \leq q} I_Q(f)$. We say $f$ is \emph{$(q,\delta)$-resilient} if $I_q(f) \leq \delta$. In addition, for $0 < \tau < 1$, we say $f$ is $\tau$-strongly resilient if for all $1 \leq q \leq n$, $I_q(f) \leq \tau \cdot q$. 
\end{definition}

Intuitively, $I_q(f)$ quantifies the amount of influence any set of $q$ variables can exert on the evaluation of the function $f$. If $f$ is almost-balanced and $I_q(f)$ is small, say $o(1)$, then evaluating $f$ gives a one-round coin-flipping protocol that outputs a nearly unbiased bit even in the presence of up to $q$ bad players. More information and discussion of other models can be found in the survey of Dodis \cite{Dodis06}. 

In their work introducing the problem, Ben-Or and Linial constructed an explicit balanced\footnote{We say $f:\zo^n \to \zo$ is balanced if $\pr_{x \in_u \zo^n}[f(x) = 1] = 1/2$.} $(1/n^\alpha)$-strongly resilient function for $\alpha = \log_3 2$. Subsequently, the seminal work of Kahn, Kalai, and Linial \cite{KahnKL88} showed that for any balanced function $f$, $I_q(f) = 1 - o(1)$ for $q = \omega(n/\log n)$. Following this, Ajtai and Linial \cite{AjtaiL93} showed the \emph{existence} of a balanced function that is $\tau$-strongly resilient for $\tau = O(\log^2 n/n)$; this in particular implies the existence of a $(\Omega(n/\log^2 n), 1/3)$ resilient function. However, the construction in Ajtai and Linial is probabilistic and does not lead to an efficiently computable resilient function---it only gives a non-uniform polynomial-size circuit for computing such a resilient function. We construct an efficiently computable function matching the existential result of Ajtai and Linial:

\begin{theorem}[Main]\label{th:main}
For some universal constants $c_1, c_2 \geq 1$ the following holds. There exists an efficiently computable function $f:\zo^n \to \zo$ such that\footnote{Henceforth, for a multi-set $S$, $x \in_u S$ denotes a uniformly random element of $S$.}
\begin{itemize}
\item $f$ is almost-balanced: $\pr_{x \in_u \zo^n}[f(x) =1 ] = 1/2 \pm 1/10$.
\item $f$ is $(c_1 (\log^2 n)/n)$-strongly resilient. 
\item $f$ has a \emph{uniform} depth $3$ 	monotone circuit of size at most $n^{c_2}$. 
\end{itemize}
\end{theorem}

The best previously known explicit resilient functions as above could only tolerate roughly at most $q \ll n^{1-\beta}$ bad players for all constants $0 < \beta < 1$: \cite{Meka09} gave such a function of logarithmic-depth while the recent breakthrough of \cite{ChattopadhyayZ15}  gave such a function of depth $4$.  

The existential guarantee of \cite{AjtaiL93} is slightly stronger than the above; they show the existence of a balanced constant-depth function with similar resilience; however, their function is not monotone\footnote{While it is possible to make their construction monotone, this blows up the depth.}. Our construction essentially matches theirs while being efficiently computable. We can have the bias of the function be $1/2 \pm o(1)$ at the expense of reducing the resilience; see \cref{cor:polyresilientkwise} for one such trade-off.

\subsection{Two-source extractors} The present work builds on a recent breakthrough of Chattopadhyay and Zuckerman \cite{ChattopadhyayZ15} who gave an explicit \emph{two-source extractor} for poly-logarithmic \emph{min-entropy} sources---resolving a longstanding problem in pseudorandomness. One of the main building blocks of their work is an efficiently computable resilient function with a stronger guarantee described below. We noticeably simplify the construction and analysis of \cite{ChattopadhyayZ15} and obtain better quantitative bounds. We explain these next starting with the definitions of extractors.

For a random variable $X$, the min-entropy of $X$ is defined by 
$$H_\infty(X) = \min_{x \in Support(X)}(\log_{2} (1/\pr[X=x])).$$
 A two-source extractor is a function that takes two high min-entropy sources and outputs a nearly uniform random bit:

\begin{definition}
A function $Ext:\zo^n \times \zo^n \to \zo$ is a $(n,k)$ \emph{two-source extractor} with error $\epsilon$ if for any two independent $X,Y$ with $H_\infty(X), H_\infty(Y) \geq k$, $Ext(X,Y)$ is $\epsilon$-close to a uniformly random bit. If $Ext(X,Y)$ has full support for all such sources $X,Y$, then we say $Ext$ is a \emph{two-source disperser}. 
\end{definition}

Extractors have many applications across several areas including cryptography, error-correcting codes, randomness amplification; we refer the reader to the history of the problem in \cite{ChattopadhyayZ15} (the references are too many). By the probabilistic method, there exist two-source extractors with error $2^{-\Omega(k)}$ for all $k \geq 2\log n$ (even outputting $\Omega(k)$ bits; as in \cite{ChattopadhyayZ15} we only focus on extracting one bit in this work). Constructing such functions explicitly, as is required in most applications, is much harder and has been studied extensively. 

Until very recently, the best explicit two-source extractor due to Bourgain \cite{Bourgain05} required min-entropy at least $k \geq c n$ for some constant $c = 0.49... < 0.5$; the best explicit two-source disperser due to \cite{BarakRSW12} required min-entropy at least $\exp(\poly(\log \log n))$. Chattopadhyay and Zuckerman broke the barrier for two-source extractors and gave an explicit construction for min-entropies at least $C (\log n)^{74}$; independently, Cohen \cite{Cohen15} gave an explicit two-source disperser for min-entropy $\log^C n$ for some (unspecified) constant $C$. We show the following: 
\begin{theorem}\label{th:ext1}
For a sufficiently big constant $C$, there exists an explicit $(n,k)$ two-source extractor $D: \zo^n \times \zo^n \to \zo$ with constant-error for $k \geq C \log^{10} n$. In particular, we get an explicit $(n,k)$ two-source disperser for $k \geq C \log^{10} n$.
\end{theorem}

\begin{theorem}\label{th:ext2}
For any constant $c \geq 1$, there exists a constant $C$ such that there exists an explicit $(n,k)$ two-source extractor $Ext: \zo^n \times \zo^n \to \zo$ with error $1/n^c$ for $k \geq C \log^{18} n$. 
\end{theorem}

\paragraph{Two-source extractors and resilient functions}
Along with the above quantitative improvements, our construction and analysis simplify \cite{ChattopadhyayZ15}:
\begin{itemize}
\item \cite{ChattopadhyayZ15} uses several known extractors in their construction of $n^{1-\beta}$-resilient functions for all constant $0 < \beta < 1$; we present a simpler construction based on Reed-Solomon codes.
\item More importantly, the analysis of \cite{ChattopadhyayZ15} uses Braverman's celebrated result---\cite{Braverman10}---that polylog-wise independence fools constant-depth circuits. We present a direct and self-contained analysis without recourse to Braverman's result that in turn uses several non-trivial tools from the study of constant-depth circuits, e.g., \cite{Razborov87}, \cite{Smolensky87}, \cite{LinialMN93}. 
\end{itemize}

Chattopadhyay and Zuckerman  reduce the problem of computing two-source extractors to that of constructing an explicit $n^{1-\delta}$-resilient function for some constant $\delta > 0$ with the following stronger property; this reduction was also implicit in \cite{Li15}. A distribution $\calD$ on $\zo^n$ is $t$-wise independent if for $X \lfta \calD$, and all $I\subseteq [n]$ with $|I| \leq t$, the projection of $X$ onto the coordinates in $I$, $X_{I},$ is uniformly distributed over $\zo^I$. 

\begin{definition}
For a Boolean function $f:\zo^n \to \zo$, $Q \subseteq [n]$, and a distribution $\calD$ on $\zo^{[n]\setminus Q}$, let $I_{Q,\calD}(f)$ be the probability that $f$ is not-determined by setting the bits not in $Q$ according to $\calD$. Let $I_{q,t}(f) = \min\{I_{Q,\calD}(f): Q\subseteq [n], |Q| \leq q,\; \text{ $\calD$ is $t$-wise independent}\}$. We say $f$ is $t$-wise \emph{$(q,\epsilon)$-resilient} if $I_{q,t}(f) \leq \epsilon$. In addition, for $0 < \tau < 1$, we say $f$ is $t$-wise $\tau$-strongly resilient if for all $1 \leq q \leq n$, $I_{q,t}(f) \leq \tau \cdot q$. 

\end{definition}

The core of \cite{ChattopadhyayZ15} is the construction of a $\poly(\log n)$-wise $(n^{1-\beta},1/n^{\Omega(1)})$-resilient function for all $\beta > 0$. They achieve this as follows: a) Construct an explicit constant-depth monotone $(n^{1-\beta},1/n^{\Omega(1)})$-resilient function. b) Apply Braverman's result to conclude that for all $q$, $\epsilon > 0$, a constant-depth monotone $(q,\epsilon)$-resilient function is also $t$-wise $(q,2\epsilon)$-resilient for $t = \poly(\log(n/\epsilon))$. In contrast, our analysis of resilience is robust enough to imply resilience even under limited independence with little extra work. As a corollary, we get the following:
\begin{theorem}\label{th:mainkwise}
For some universal constant $c \geq 1$ the following holds. There exists an efficiently computable function $f:\zo^n \to \zo$ such that,
\begin{itemize}
\item $f$ is almost-balanced: For any $(c \log^2 n)$-wise independent distribution $\calD$ on $\zo^n$, $\pr_{x \sim \calD}[f(x) =1 ] = 1/2 \pm 1/9$.
\item $f$ is $(c \log^2 n)$-wise $(c (\log^2 n)/n)$-strongly resilient. 
\item $f$ has a \emph{uniform} depth $3$ 	monotone circuit of size at most $n^{c}$. 
\end{itemize}
\end{theorem}

As in the case of \tref{th:main}, we can have the bias of the function be $1/2 \pm o(1)$ at the expense of reducing the resilience.

\subsection{Oblivious samplers preserving moment generating functions} 
A critical ingredient in our proof of \tref{th:main} is an explicit \emph{oblivious sampler} with optimal---up to constant factors---seed-length that approximates \emph{moment generating functions} (MGF). We state this result next which may be of independent interest.

\begin{theorem}\label{th:expsampler}
For all $0 <  \mu \leq 1$, $1 \leq v,w$, there exists an explicit generator $G:\zo^r \to [v]^w$ such that for all functions $f_1,\ldots,f_w:[v] \to [0,1]$ with $\sum_{i=1}^w \E_{x \in_u [v]}[f_i(x)] = \mu$, 
$$\E_{y \in_u \zo^r}\sbkets{2^{\sum_{i=1}^w f_i(G(y)_i)}} = 1 + O(\mu).$$
The seed-length of the generator is $r = O(w + (\log v) + w((\log \log v) + \log(1/\mu))/(\log w))$. 
\end{theorem}
We state a more precise version which works for estimating $\E\sbkets{\exp\bkets{\theta \cdot \sum_i f_i(\;)}}$ for all $\theta > 0$, i.e., the MGF, in \sref{sec:expsampler}; here we focus on the above for simplicity and as it captures the main ideas. 

Let us first compare the above with known randomness efficient samplers such as those of \cite{Zuckerman97, Gillman98, Kahale97}. For concreteness, let us consider the special case when $v = 2^{O(w)}$ and $\mu = 1$; these are the parameters we face in our application. The seed-length of \tref{th:expsampler} in this case is $O(w)$. On the other hand, the samplers of \cite{Zuckerman97, Gillman98, Kahale97} when instantiated to obtain a guarantee as above require a seed-length of $\Omega(w \log w)$. The improvement from $O(w \log w)$ to $O( w)$ is critical in our application; as we enumerate our all possible seeds eventually, the improvement from $O(w \log w)$ to $O(w)$ in seed-length translates to an improvement from super-polynomial running-time ($n^{O(\log \log n)}$) to polynomial running-time of the resilient function. 

To illustrate the gap further, with the above parameters, let $X_i = f_i(x_i)$ for an independent uniformly random $x_i \in_u [v]$ and let $Y_i = f_i(G(y)_i)$, where $G$ is as in the theorem. In this case, standard Chernoff bounds imply that for a sufficiently big constant $C$,
\begin{equation}\label{eq:truechernoff}
\pr[ X_1 + \cdots + X_w > C w/(\log w)] \leq \exp(-w).
\end{equation}

Our argument in \sref{sec:expsampler} shows that our generator satisfies the same property: $\pr[Y_1 + \cdots + Y_w > C w/(\log w)] \leq  \exp(-w)$. Note that the seed-length used by our generator for this setting is $O(w)$ which is optimal\footnote{To get a tail bound of $\exp(-w)$, the sample space has to have $\exp(w)$ points.}.

In contrast, if one uses the \emph{expander sampler} as in \cite{Gillman98, Kahale97, Healy08} on an expander graph with degree $D$ the best one could get using the current analyses is (cf.~Corollary 23 of \cite{Healy08}),
$$\pr[Y_1 + \cdots + Y_w > C w/(\log w)] \leq \exp(-w \cdot (1 - (w^2/D^4))).$$
In other words, to get an $\exp(-w)$ bound on the tail-bound, one needs the degree $D$ of the expander to be $w^{\Omega(1)}$; as the number of random bits needed by such a generator is $(\log v) + w (\log D)$, the total seed-length will be $O(w \log w)$. Similarly, applying the analysis of expander Chernoff bounds from \cite{Gillman98, Kahale97, Healy08} to get a bound on the MGF as in the theorem requires the degree of the expander to be at least $w^{\Omega(1)}$. This in turn requires the total seed-length to be $\Omega(w \log w)$. 

Another potential approach for constructing samplers as in theorem is to apply the pseudorandom generators (PRGs) for small-space machines of Nisan or Impagliazzo, Nisan, and Wigderson \cite{Nisan92, ImpagliazzoNW94}.  However, to obtain a guarantee on the MGF as in the theorem or to satisfy \eref{eq:truechernoff} in the special-case above, we need to instantiate the generators with error $\ll \exp(-w)$. This in turn forces the seed-length of the generators to be $\Omega(w \log w)$. 

\section{Overview of construction}
We next give a high level overview of our main construction and analysis. First, some notations:
\begin{itemize}
\item Throughout, by a partition $P$ of $[n]$ we mean a division of $[n]$ into $w$-sized blocks $P_1,\ldots,P_v$, where $v = n/w$ (we assume $w$ divides $n$). 
\item Let $\calD_p$ denote the product distribution on $\zo^n$ where each bit is $p$-biased. 
\end{itemize}

As in \cite{AjtaiL93} and \cite{ChattopadhyayZ15} our construction will be an \emph{AND} of several \emph{Tribe} functions:
\begin{definition}
For a partition $P = \{P_1,\ldots, P_v\}$ of $[n]$ into $w$-sized blocks, the associated \emph{Tribes} function is the DNF defined by $T_P = \vee_{j=1}^v \bkets{\wedge_{\ell \in P_j} x_\ell}$. A collection of partitions $\calP = \{P^1,\ldots,P^u\}$ defines a function $f \equiv f_{\calP}:\zo^n \to \zo$ as follows:
$$f_\calP(x) := \bigwedge_{i=1}^u \;\bigvee_{j=1}^v \;\bkets{\bigwedge_{k \in P^i_j} x_k} =  \bigwedge_{i=1}^u T_{P^i}(x).$$
\end{definition}

The final function satisfying \tref{th:main} will be $f_\calP$ for a suitably chosen set of partitions. To analyze such functions, we first state two abstract properties that allow us to analyze the bias as well as influences of such functions; we then design partitions that satisfy the properties. The properties we define are motivated by \cite{ChattopadhyayZ15} and abstracting them in this way allows us to give a modular analysis of the construction. The partitions themselves will be designed using the sampler we construct in \tref{th:expsampler} which forms the core of our analysis and construction. The analysis of resilience under limited independence follows a similar approach in addition to some careful, but elementary, calculations involving elementary symmetric polynomials. 

\paragraph{Analyzing bias} The first condition allows us to approximate the bias of functions of the form $f_\calP$.

\begin{definition}
Let $\calP = \{P^1,\ldots,P^u\}$ be a collection of partitions of $[n]$ into $w$-sized blocks. For $d \leq w$, we say $\calP$ is a $d$-design if no two blocks across any of the partitions overlap in more than $w-d$ elements: formally, for all $\alpha \neq \beta \in [u]$, and $i,j \in [v]$, $|P^\alpha(i) \cap P^\beta(j)| \leq w-d$. In addition, for $d \leq k \leq w$ and $\delta \in (0,1)$, we say $\calP$ is a $(d,k,\delta)$-design if it is a $d$-design and for all $\alpha \in [u]$, and $i,j \in [v]$,  
$$\pr_{\beta \in_u [u]}\sbkets{|P^\alpha_i \cap P^\beta_j| \leq w-k} \geq 1 - \delta.$$
\end{definition}
We should think of $k \gg d$. Intuitively, the first condition says that any two blocks arising in our partitions differ in at least $d$ elements (i.e., do not overlap completely); in contrast, the second condition says that with probability $1-\delta$, two random blocks differ in at least $k$ elements (i.e, have very little overlap if $k \gg d$). 

When a collection of partitions $\calP = \{P^1,\ldots,P^u\}$ satisfies the above condition, the following claim gives a formula for the bias of $f_\calP$. To parse the formula, even if clearly false, suppose that the tribes involved in $f_\calP$ were on disjoint sets of variables. Then, we would have
$$\pr_{x \in_u \zo^n}[f_\calP(x) = 1] = \prod_{\alpha=1}^u \pr_{x \in_u \zo^n}\sbkets{T_{P^\alpha}(x) = 1} = (1 - (1- 2^{-w})^v)^u := \mathsf{bias}(u,v,w).$$
The next claim shows that when $\calP$ forms a $(d,k,\delta)$-design, the different tribes behave as if they are on disjoint sets of variables and the above formula for the bias is approximately correct. As is sufficient in our applications, we specialize to the case where $\mathsf{bias}(u,v,w)$ is close to $1/2$; this corresponds to choosing\footnote{Our arguments also give analogous, albeit more cumbersome, bounds for all $u,v,w$ and even other product distributions.} $v = \Theta(1) w 2^w$ and $u$ such that $1/3 \leq \mathsf{bias}(u,v,w) \leq 2/3$. 
\begin{theorem}\label{th:biascontrol}
Let $\calP = \{P^1,\ldots,P^u\}$ be a collection of partitions of $[n]$ into $w$-sized blocks that is a $(d,k,\delta)$-design. Let $u,v,w$ be such that $v = \Theta(1) w 2^w$ and $1/3 \leq \mathsf{bias}(u,v,w) = (1 - (1- 2^{-w})^v)^u \leq 2/3$. Then,
$$\abs{\pr_{x \in_u \zo^n}[f_\calP(x) = 1] - \mathsf{bias}(u,v,w)} \leq \min\begin{cases}O(1) w \exp(-\Omega(d))\\ O(1) \cdot \bkets{w \exp(-\Omega(k)) + \exp(-\Omega(d)) + 2^{w} \delta} \end{cases}. $$
\end{theorem}

\ignore{
When $\calP$ satisfies the above condition, the following claim gives a formula for the bias of $f_\calP$. We state the result for general $p$-biased distributions as the statement and its analysis are no more complicated. 

\begin{theorem}\label{th:biascontrol}
Let $\calP = \{P^1,\ldots,P^u\}$ be a collection of partitions of $[n]$ into $w$-sized blocks that is a $(d,k,\delta)$-design for some $d \leq w/2$. Let $0 \leq p \leq 1$ with $v \geq p^{-w}$ and $\theta = (1-p^w)^v$. Then, for all even integers $r \leq u$,  
$$\pr_{x \lfta \calD_p}[f_\calP(x) = 1] = (1-(1-p^w)^v)^u \pm O(1) \bkets{ \bkets{v^2 r^2\exp(2 v r^2 p^{w+d})} \cdot \delta + (v r^2) \cdot  p^{2w-k}}  \cdot (1 + \theta)^u \pm 2 \theta^r \binom{u}{r}.$$
\end{theorem}}

The proof of \tref{th:biascontrol} is similar to the arguments in \cite{ChattopadhyayZ15} and relies on Janson's inequality. However, our argument is more subtle as we need to handle $(d,k,\delta)$-designs and not just $d$-designs as is done there. In particular the theorem implies that when $\calP$ is a $d$-design with $d \gg \log w$, the error is at most $\exp(-\Omega(d))$. On the other hand, when $d \ll \log w$ we can use the second formula if $k \gg \log w$. We need the more refined statement above where only most blocks are far from each other as this is what our construction achieves. 

\paragraph{Analyzing influences} We next specify a sufficient condition on a collection of partitions $\calP = \{P^1,\ldots,P^u\}$ to guarantee that small coalitions have small influence on $f_\calP$. 
\begin{definition}
Let $\calP = \{P^1,\ldots,P^u\}$ be a collection of partitions of $[n]$ into $w$-sized blocks. We say $\calP$ is $(q,\tau)$-load balancing if for all $Q \subseteq [n]$ with $|Q| \leq q$, and $j \in [v]$,
$$\E_{\alpha \in_u [u]}\sbkets{\mathsf{1}(Q \cap P^\alpha_j \neq \emptyset) 2^{|Q \cap P^\alpha_j|}} \leq \tau \cdot (q/v).$$
\end{definition}

To gain some intuition for the definition and the use of the name \emph{load balancing} we first view partitions as hash functions. A partition $P = \{P_1,\ldots, P_v\}$ can be seen as a hash function $h_P:[n] \rightarrow [v]$: $h_P(i) = j$ if $i \in P_j$. We can then also view a collection of partitions $\calP$ as defining a family of hash functions $\cal{H} = \{h_P: P \in \calP\}$. The definition above then says that a random hash function from the family $\cal{H}$ is \emph{load-balancing} in a certain concrete way. In particular, for any subset $Q \subseteq [n]$, and any particular bin $j \in [v]$, the number of items hashed into bin $j$ satisfy the following inequality:
$$\E_{h \in_u \cal{H}}\sbkets{ \mathsf{1}(Q \cap h^{-1}(j) \neq \emptyset) 2^{|Q \cap h^{-1}(j)|}} \leq \tau \cdot (|Q|/v).$$

Indeed, standard Chernoff bounds imply that the above property is clearly satisfied by truly random hash functions. Thus, the partitions we construct can be seen as imitating this property of truly random hash functions but with a much smaller size family. 

We show that load-balancing partitions give us fine control on the influences:

\begin{theorem}\label{th:influencecontrol}
Let $\calP = \{P^1,\ldots,P^u\}$ be a collection of partitions of $[n]$ into $w$-sized blocks that is $(q,\tau)$-load balancing. Then,
$$I_{q}(f_\calP) \leq (u (1-2^{-w})^{v-q}) \cdot (\tau 2^{-w}) \cdot q.$$
\end{theorem}

A similar claim is used in the analysis of \cite{AjtaiL93, ChattopadhyayZ15}. However, \cite{ChattopadhyayZ15} work with the stronger condition that for any $Q \subseteq [n]$ with $|Q| \ll q$, $|Q \cap P^\alpha_j| \ll w$ for most $\alpha$ (as opposed to just having a bound on the expectation of $2^{|Q \cap P^\alpha_j|}$). The above generalization, while straightforward, is important as one cannot hope to satisfy their stronger requirement  for $Q$ very large ($n^{1-o(1)}$) as needed for the proof of \tref{th:main}. 

\subsection{Constructing nice partitions}
We next outline how to construct a collection of partitions which is a \emph{good} design as well as \emph{sufficiently} load-balancing. Fix $v, w$. For a string $\alpha \in [v]^w$, define an associated partition $P^\alpha$ of $[n] \equiv [v w]$ into $w$-sized blocks as follows: 
\begin{itemize} 
\item Write $\{1,\ldots v w\}$ from left to right in $w$ blocks of length $v$ each. Now, permute the $k$'th block by shifting the integers in that block by adding $\alpha_k$ modulo $v$.
\item The $i$'th part now comprises of the elements in the $i$'th position in each of the $w$ blocks.
\end{itemize}
Formally, for $i \in [v]$ $P^\alpha_i = \{(k-1)v + ((i-\alpha_k) \mod v): k \in [w]\}$. As in \cite{ChattopadhyayZ15}, our final function will be $f_\calP$ for $\calP:= \calP_\calU = \{P^\alpha: \alpha \subseteq \calU\}$ for a suitably chosen set of strings $\calU \subseteq [v]^w$; in their work, $\calU$ is chosen by using known \emph{seeded} extractors. 

For intuition, fix a constant $0 < \beta < 1$ and first consider the case where $\calU$ is the set of Reed-Solomon codewords corresponding to degree $\ell \geq 1/\beta$ polynomials over $[v]$ \footnote{Assuming for simplicity that $v$ is a prime.}. Then, a simple calculation shows that $\calP_\calU$ is a $(w-\ell)$-design; that is, no two blocks in the partitions of $\calP$ overlap in more than $\ell$ positions. Further, note that a random element of $\calU$ is $\ell$-wise independent; combining this with standard Chernoff-type bounds for $\ell$-wise independent hash functions implies that $\calP$ is $(q,O_\beta(1))$-load balancing for $q \ll n^{1-\beta}$. Setting the parameters appropriately and applying Theorems \ref{th:biascontrol} and \ref{th:influencecontrol} shows that $f_\calP$ is almost-balanced and $(n^{1-\beta})$-resilient---giving a simpler construction matching \cite{ChattopadhyayZ15}. 

For the main theorem, \tref{th:main}, we follow the same outline as above using \tref{th:expsampler} to get a suitable set of partitions. Indeed, a look at the definition of load-balancing suggests that \tref{th:expsampler} should be relevant as it also involves a similar expression. Concretely, as a first attempt, we take $\calU$ to be the range of the generator from \tref{th:expsampler} and study the bias and influence of $f \equiv f_{\calP_\calU}$. \tref{th:expsampler} immediately implies that $\calP_\calU$ is $(\Omega(n/\log n), O(1))$-load balancing. This combined with \tref{th:influencecontrol} implies that $f$ is $O((\log^2 n)/n)$-strongly resilient.

To complete the proof, we need to show that $f$ is almost-balanced. Unfortunately, $\calP_\calU$ need not be a good-enough design to apply \tref{th:biascontrol} directly. We get around this at a high-level by encoding parts of the output of the oblivious sampler in \tref{th:expsampler} using a Reed-Solomon code; doing so, we get $\calP_\calU$ to be a $(C, w/2,2^{-Cw})$-design for any large constant $C$. We then apply \tref{th:biascontrol} to show that $f$ is almost-balanced. We leave the details of the encoding to the actual proof. 

\subsection{Construction of the oblivious sampler}
The generator is obtained by instantiating the Nisan-Zuckerman \cite{NisanZ96} PRG for small-space machines with $k$-wise independent seeds being fed into the \emph{extractor} rather than truly independent ones. Concretely, let $E:[v^c] \times [D] \to [v]$ be a $(c(\log v)/2,\epsilon)$-extractor (see \sref{sec:prelims} for formal definitions) with error $\epsilon \approx 1/w$. For $\ell = \Theta(w/(\log w))$, let $G_\ell:[D]^\ell \to [D]^w$ generate a $\ell$-wise independent distribution. Then, our generator satisfying \tref{th:expsampler}, $G:[v^c] \times [D]^\ell \to [v]^w$ is defined as follows:
\begin{equation}\label{eq:expintro}
G(x,y) = \bkets{E(x,G_\ell(y)_1), E(x,G_\ell(y)_2), \cdots, E(x,G_\ell(y)_w)}.
\end{equation}
 
While the above construction serves as the base, in our proof of \tref{th:main} we need a generator that satisfies certain additional constraints: output strings on different seeds should be far from each other. We satisfy these constraints by careful modifications of the above construction. 

\subsection{Analyzing bias and resilience under limited independence}
 The extension to obtain resilient functions under limited independence, as in \tref{th:mainkwise}, follows from our analysis of \tref{th:biascontrol} along with the following claim. The core of the proof of \tref{th:biascontrol}  involves calculating the biases of disjunctions of a small number of Tribes, i.e., functions of the form $h = \vee_{\alpha \in I} T_{P^\alpha}$, where $P^\alpha$ are partitions of $[n]$ and $|I|$ is not too large. We directly show that  such functions are \emph{fooled} to within error $\epsilon$ by $O(|I|^2 w^2 + |I| w \log(1/\epsilon))$-wise independence; that is, for any $O(|I|^2 w^2 + |I| w \log(1/\epsilon))$-wise independent distribution $\calD$ on $\zo^n$,
$$\abs{\pr_{x \in_u \zo^n}[ \vee_{\alpha \in I} T_{P^\alpha}(x) = 1] - \pr_{x \sim \calD}[ \vee_{\alpha \in I} T_{P^\alpha}(x) = 1]} < \epsilon.$$
The proof of the above claim involves standard approximations based on the inclusion-exclusion principle and some inequalities involving elementary symmetric polynomials.

\section{Preliminaries}\label{sec:prelims}
\subsection{Pseudorandomness}
We first recall some standard notions from pseudorandomness.
\begin{definition}
A collection of random variables $\{X_1,\ldots,X_m\}$ is $k$-independent if for any $I \subseteq [m], |I| \leq k$, the random variables $\{X_i: i \in I\}$ are independent of each other. 
\end{definition}

\begin{definition}
For a distribution $X$, $H_\infty(X) = \min_{x \in Support(X)} \log(1/\pr[X=x])$. 
\end{definition}

\begin{definition}[\cite{NisanZ96}]
A function $E:[N] \times [D] \to [M]$ is a $(k,\epsilon)$-\emph{strong extractor} if for every distribution $X$ over $[N]$ with $H_\infty(X) \geq k$, and $Y \in_u [D]$, $(Y, E(X,Y))$ is $\epsilon$-close in statistical distance to the uniform distribution over $[D] \times [M]$.
\end{definition}

We need the following sampling properties of strong extractors, c.f., \cite{Zuckerman97}.
\begin{lemma}\label{lm:extsampler}
Let $E:[N] \times [D] \to [M]$ be a $(k,\epsilon)$-strong extractor. Then, for all functions $g_1,\ldots,g_D:[M] \to \zo$, with $\mu = (1/D) \sum_i \E_{x \in_u [M]}[g_i(x)]$, there are at most $2^k$ elements $x \in [N]$ such that
$$\abs{\frac{1}{D} \sum_{z \in [D]}  g_z (E(x,z)) - \mu} \geq \epsilon.$$
\end{lemma}

We also need the following explicit extractor construction due to Zuckerman \cite{Zuckerman97}:
\begin{theorem}\label{th:zextractor}
There exists a constant $C \geq 1$ such that for all $\epsilon > 0$, and $ 1 \leq M \leq N^{1/3}$, there is an explicit $((\log N)/2, \epsilon)$-strong extractor $E:[N] \times [D] \to [M]$ with $D = ((\log N)/\epsilon)^C$. We also assume without loss of generality that $E(X,Y)$ is uniformly random over $[M]$ when $X,Y$ are uniformly random over $[N]$ and $[D]$ respectively. 
\end{theorem}

Finally, we also need the following explicit generator of $\ell$-wise independent distributions which follows for instance from  the \emph{Reed-Solomon} code.
\begin{definition}[Shift-Hamming distance]
For two sequences $x,x'$ over an alphabet $[B]^d$, let $d_H(x,x') = \min_{a \in [B]} |\{i \in [d]: x_i - x_i' \neq a \mod B\}|$. 
\end{definition}

\begin{lemma}\label{lm:kwisecode}
For all prime $v$ and $1 \leq \ell \leq m \leq v$, there exists an explicit function $G_\ell:[v]^\ell \to [v]^m$ such that 
\begin{itemize}
\item For any $y \neq y' \in [v]^\ell$, $d_H(y,y') \geq m - \ell$.
\item For $y \in_u [v]^\ell$, $G_\ell(y)$ is $\ell$-wise independent.
\end{itemize} 
\end{lemma}
\begin{proof}
Follows by using the Reed-Solomon code over $[v]$ of degree $\ell$ and length $m$. 
\end{proof}

We need the following theorem of \cite{DeETT10} on fooling read-once Conjunctive Normal Formulas (CNFs) using limited independence. A CNF $f:\zo^n \to \zo$ is said to be a readonce formula if each variable appears in at most one clause. The width of a CNF is the maximum length of any clause in the CNF. 
\begin{theorem}\label{th:rcnf}
There exists a constant $C$ such that the following holds for all $0 < \epsilon < 1$ and $t \geq C w \log(1/\epsilon)$: For any width $w$ read-once CNF $f:\zo^n \to \zo$ and $t$-wise independent distribution $\calD$ on $\zo^n$, 
$$\abs{\pr_{x \lfta \calD}[f(x) = 1] - \pr_{x \in_u \zo^n}[f(x) = 1]} \leq \epsilon.$$
\end{theorem}

\subsection{Probability}
We next review some results from probability theory. 

We need Janson's inequality from probability theory, c.f., \cite{AlonS11} that allows us to bound the bias of \emph{monotone Disjunctive Normal Formulas (DNFs)}. 
\begin{theorem}\label{th:janson1}
Let $S_1,\ldots,S_m \subseteq [n]$ be a collection of sets and define $f_i:\zo^n \to \zo$ by $f_i(x) = \wedge_{j \in S_i} x_j$. Let $x \in_u \zo^n$. Then,
$$\prod_{i=1}^m \pr\sbkets{\neg f_i(x) = 1} \leq \pr\sbkets{\bigwedge_{i=1}^m (\neg f_i(x)) = 1} \leq \exp\bkets{\frac{\Delta}{1-\gamma}} \prod_{i=1}^m \pr\sbkets{\neg f_i(x) = 1},$$
where $\gamma = \max_{i=1}^m \pr[f_i(x) = 1]$ and 
$$\Delta = \sum_{i \neq j: S_i \cap S_j \neq \emptyset} \pr[f_i(x) \wedge f_j(x) = 1].$$
\end{theorem}
The following immediate corollary is more convenient for us. 
\begin{theorem}\label{th:janson}
Let $S_1,\ldots,S_m \subseteq [n]$ be a collection of sets and define $f_i:\zo^n \to \zo$ by $f_i(x) = \wedge_{j \in S_i} x_j$. Let $x \in_u \zo^n$, and $Z_i = f_i(x)$ for $i \in [m]$. Then,
$$\prod_{i=1}^m \E[(1 - Z_i)] \leq \E\sbkets{\prod_{i=1}^m \bkets{1-Z_i}} \leq \exp\bkets{\frac{\Delta}{1-\gamma}} \cdot \prod_{i=1}^m \E[(1 - Z_i)],$$
where $\gamma = \max_{i=1}^m \E[Z_i]$ and 
$$\Delta = \sum_{i \neq j: S_i \cap S_j \neq \emptyset} \E[Z_i Z_j].$$
\end{theorem}

We need the following instantiation of Holder's inequality.
\begin{fact}\label{fct:Holder}
Let $Y_1,\ldots,Y_r$ be real-valued random variables. Then, 
$$\E\sbkets{|Y_1 Y_2 \cdots Y_r|} \leq \prod_{i=1}^r \E\sbkets{|Y_i|^r}^{1/r}.$$
\end{fact}
\begin{definition}
For any collection of variables $x_1,\ldots,x_m$ and $1 \leq a \leq m$, $S_a(x_1,\ldots,x_m) = \sum_{I \subseteq [m], |I| = a} \prod_{i \in I} x_i$ denotes the $a$'th symmetric polynomial.
\end{definition}

We will use the following bound on moments of symmetric polynomials of independent Bernoulli random variables; see the appendix for proof. 
\begin{lemma}\label{lm:binomialmoments}
Let $0 < p < 1/2$ and $X_1,\ldots,X_v$ be independent indicator random variables with $\pr[X_i = 1] = p$. Then, for $k \geq 2 e^2 v p$, and all $r \geq 1$, 
$$\E[S_k(X_1,\ldots,X_v)^r] \leq (1/2)^{k}.$$
\end{lemma}

We will use the following standard fact about symmetric polynomials.

\begin{fact}\label{fact:symmax} 
For all $1 \leq a \leq m$ and $0 \leq q_1,\ldots,q_m$ with $\sum_i q_i \leq \mu$, $S_a(q_1,\ldots,q_m) \leq \binom{m}{a} \cdot (\mu/m)^a$. 
\end{fact}

We need the following elementary approximations; see the appendix for proofs. 
\begin{fact}\label{fact:expapprox} 
For all $x \geq 2$, $e^{-1}(1 - 1/x) \leq (1-1/x)^x \leq e^{-1}$.
\end{fact}

\begin{fact}\label{fact:biasapprox}
Let $1 \leq w \leq v \leq u$, $B \geq 1$ with $0 \leq v - 2^w (\ln (u/\ln 2)) \leq B$, and $\theta = (1-2^{-w})^v$. Then, $(1 + \theta)^u = O(1)$ and $\bkets{1 - \theta}^u = 1/2 \pm O(B \ln u) 2^{-w}$. 
\end{fact}

\section{Analyzing bias}
Here we prove \tref{th:biascontrol}. Let $\calP = \{P^1,\ldots,P^u\}, f_\calP$, be as in the theorem statement. Let $x \in_u \zo^n$ and for $\alpha \in [u]$ let $Z_\alpha = T_{P^\alpha}(x)$. Then, $$f_\calP(x) = \prod_{\alpha \in [u]} T_{P^\alpha}(x) = \prod_{\alpha \in [u]} Z_\alpha.$$

The theorem essentially states that $\E[\prod_\alpha Z_\alpha] \approx \prod_\alpha \E[Z_\alpha]$. We will prove this by showing that $Z_\alpha$'s satisfy a weak notion of limited independence. Roughly speaking, we will show that for $I \subseteq [u], |I| \leq d$, 
\begin{equation}\label{eq:bias00}
\E\sbkets{ \prod_{\alpha \in I} (1 - Z_\alpha)} \approx \prod_{\alpha \in I} \E[(1 - Z_\alpha)].
\end{equation}
The proof of \eref{eq:bias00} uses Janson's inequality and the design properties of $\calP$. The theorem then follows by writing 
$$\prod_\alpha Z_\alpha= \prod_\alpha (1 - (1-Z_\alpha)) = \sum_{a=1}^u S_a(\{1-Z_\alpha: \alpha \in [u]\})$$ 
and applying Bonferroni inequalities to truncate the latter expansion to the first $d$ terms; the error from the truncation is bounded by applying \eref{eq:bias00}. We next gives a concrete form of \eref{eq:bias00}.
\begin{lemma}\label{lm:bias01}
Let $\calP$ be as in \tref{th:biascontrol} and $\{Z_\alpha\}$'s be as defined above. 
\ignore{Then, for all $I \subseteq [u]$ with $|I| \leq d$,
$$\prod_{\alpha \in I} \E[( 1 - Z_\alpha)] \leq \E\sbkets{\prod_{\alpha \in I} ( 1 - Z_\alpha)} \leq \exp(O(w^2 2^{-\Omega(d)})) \cdot \prod_{\alpha \in I} \E[( 1 - Z_\alpha)] .$$}
Then, for all $a \leq d$, 
\begin{equation*}
\binom{u}{a} \cdot (1 - 2^{-w})^{va}  \leq \E\sbkets{S_a(1- Z_1, 1-Z_2,\ldots,1-Z_u)} \leq \gamma \cdot \binom{u}{a}  \cdot (1 - 2^{-w})^{va}, 
\end{equation*}
where $\gamma = \exp\bkets{ O(1) w d 2^{-\Omega(k)}} + 2^{O(w)} \delta$. 
\end{lemma}
Note that $\E[Z_\alpha] = (1-2^{-w})^v$ for all $\alpha \in [u]$. Therefore, if the $Z_\alpha$'s were actually independent of each other, then $\E\sbkets{S_a(1- Z_1, 1-Z_2,\ldots,1-Z_u)} = \binom{u}{a} (1 - 2^{-w})^{va} $. Thus, the lemma says that for a design, the expectation of the symmetric polynomial is close to what it would be if $\{Z_\alpha\}$'s were independent of each other.

Before proving the lemma, we first show it implies \tref{th:biascontrol}. 
\begin{proof}[Proof of \tref{th:biascontrol}]
Let $Z_\alpha$'s be as above; for brevity, let $Y_\alpha = 1 - Z_\alpha$ and $\theta = (1-2^{-w})^v$. Note that
\begin{equation}\label{eq:biasc1}
\pr[f_\calP(x) = 1] = \E\sbkets{\prod_{\alpha \in [u]} Z_\alpha},\;\;\;\;\;\;\\\prod_{\alpha \in [u]}\E[Z_\alpha] = (1- \theta)^u = \mathsf{bias}(u,v,w).
\end{equation}

Now, by Bonferroni inequalities, we have
$$\prod_\alpha Z_\alpha = \prod_\alpha (1-Y_\alpha) = \sum_{a=0}^{d-1} (-1)^t S_a(Y_1,\ldots,Y_u) \pm S_d(Y_1,\ldots,Y_u).$$
Thus, 
$$\abs{\E\sbkets{\prod_\alpha Z_\alpha} - \sum_{a=0}^{d-1} (-1)^t \E[S_a(Y_1,\ldots,Y_u)]} \leq \E[S_d(Y_1,\ldots,Y_u)].$$

Further, by \lref{lm:bias01}, for all $1 \leq a \leq d$, 
$$\binom{u}{a} \theta^a \leq \E[S_a(Y_1,\ldots,Y_u)] \leq \gamma  \binom{u}{a} \theta^a .$$

Let $Y_1',\ldots,Y_u'$ be independent random variables with $Y_i'$ having the same marginal as $Y_i$. Clearly, $\E[S_a(Y_1',\ldots,Y_u')] = \binom{u}{a} \theta^a$. Therefore, for any $a \leq d$, 
$$\abs{\E[S_a(Y_1,\ldots,Y_u)] - \E[S_a(Y_1',\ldots,Y_u')]} \leq (\gamma-1) \binom{u}{a} \theta^a.$$

Combining the above inequalities, we get
\begin{align*}
\abs{\E\sbkets{\prod_i Z_i} - \sum_{a=0}^{d-1} (-1)^t \E[S_a(Y'_1,\ldots,Y'_u)]} &\leq \E[S_d(Y_1,\ldots,Y_u)] + (\gamma-1) \sum_{a=1}^{d-1} \binom{u}{a} \theta^a\\
&\leq \gamma \binom{u}{d} \theta^d + (\gamma -1) (1+\theta)^u.
\end{align*}
Note that the above arguments also apply to the case when $Z_i$'s were truly independent of each other with $\gamma = 1$. Therefore,
\begin{equation*}
\abs{ \E\sbkets{Z_1 \cdots Z_u} - \prod_{i\in [u]} \E[Z_i]} \leq 2 \gamma \binom{u}{d} \theta^d + 2 (\gamma-1) (1+\theta)^u .
\end{equation*}

We next simplify the error bound by plugging in the values of $\gamma$ and $\theta$. Recall that $\mathsf{bias}(u,v,w) = (1 - \theta)^u \in [1/3,2/3]$ so that $(1+\theta)^u = O(1)$ and $\theta = O(1/u)$. Therefore, 
\begin{equation}\label{eq:bias02}
\abs{\E[\prod_{\alpha \in u} Z_\alpha] - \prod_{\alpha \in u} \E[Z_\alpha]} = \exp(-\Omega(d)) + O(\gamma - 1) .
\end{equation}

Note that $1$ is a trivial upper bound on the left-hand side; further, a simple calculation shows that
\begin{align*}
\min(1,\gamma -1) &= \min\bkets{1,  \exp\bkets{ O(wd) 2^{-\Omega(k)}}  + 2^{O(w)} \delta- 1} \\
&\leq \min\bkets{1,  \exp\bkets{ O(wd) 2^{-\Omega(k)}} - 1} + 2^{O(w)} \delta\\
&= O(wd) 2^{-\Omega(k)} + 2^{O(w)} \delta = O(w) 2^{-\Omega(k)} + 2^{O(w)} \delta,
\end{align*}
where we used the fact that for all $\lambda > 0$ $\min(1,e^\lambda - 1) = O(\lambda)$ and $d 2^{-\Omega(k)} \leq k 2^{-\Omega(k)} = O\bkets{2^{-\Omega(k)}}$. 

Combining the above with \eref{eq:biasc1} and \eref{eq:bias02} we get
$$\abs{\pr[f_\calP(x) = 1] - \mathsf{bias}(u,v,w)} \leq \exp(-\Omega(d)) + O(w) \exp(-\Omega(k)) + 2^{O(w)} \delta.$$

This proves the second bound of the theorem. The first bound is a special case as any $(d,k,\delta)$-design is also a $(d,d,0)$-design by definition. 
\end{proof}

\begin{proof}[Proof of \lref{lm:bias01}]
We first further break up each $Z_\alpha$ according to the blocks of the partition $P^\alpha$; for each $i \in [v]$, let $Z_{\alpha i} = \wedge_{j \in P^\alpha_i} x_j$ (recall $x \in_u \zo^n$). Then, $Z_\alpha = \vee_{i \in [v]} Z_{\alpha i}$ so that $1-Z_\alpha = \prod_{i \in [v]} (1- Z_{\alpha i})$. Therefore, for any $I \subseteq [u]$, 
$$\prod_{\alpha \in I} ( 1 - Z_\alpha) = \prod_{\alpha \in I} \prod_{i \in v} (1 - Z_{\alpha i}).$$
Note that we can apply \tref{th:janson} to the latter product. Fix a set $I \subseteq [u]$. Define $(\alpha,i) \neq (\beta,j)$ to be adjacent, $(\alpha, i) \sim (\beta,j)$, if $P^\alpha_i \cap P^\beta_j \neq \emptyset$. Let, 
$$\Delta_I = \sum_{(\alpha, i) \sim (\beta, j) \in I \times [v]} \E[Z_{\alpha i} \cdot Z_{\beta j}].$$

We next bound $\Delta_I$. Let $d_I = w - \max\{ |P^\alpha_i \cap P^\beta_j|\,:\, \alpha \neq \beta \in I, i,j \in [v]\}$ quantify the maximum overlap among any two blocks of the partitions $P^\alpha, \alpha \in I$. Then, as $\calP$ is a $d$-design, $d_I \geq d$ for all $I$. Fix $\alpha \neq \beta \in I$ and an in index $i \in [v]$, and let $j_1,\ldots,j_b \in [v]$ be the indices such that $(\alpha,i) \sim (\beta,j)$, and let $w_\ell = |P^\alpha_i \cap P^\beta_{j_\ell}|$. Then, $1 \leq w_1,\ldots,w_b \leq w- d_I$. Now, 
\begin{align*}
\sum_{j \in [v]: (\beta,j) \sim (\alpha,i)} \E[Z_{\alpha i} Z_{\beta j}] &= \sum_{\ell=1}^b 2^{-|P^\alpha_i \cup P^\beta_{j_\ell}|}= \sum_{\ell=1}^b 2^{-(2w - w_\ell)} = 2^{-2w} \cdot \sum_{\ell=1}^b 2^{w_\ell}.
\end{align*}
As $w_i \in [w-d_I]$, and $\sum_i w_i = w$, the above expression is maximized by setting as many of the $w_i$'s to $w-d_I$ as possible. Thus,
$$\sum_{\ell=1}^b 2^{w_\ell} \leq \lceil w/(w-d_I) \rceil 2^{w-d_I} = 2^w \lceil w/(w-d_I) \rceil 2^{-d_I} \leq 2^{w+1} 2^{-d_I/2}.$$

Therefore,  
$$\sum_{j \in [v]: (\beta,j) \sim (\alpha,i)} \E[Z_{\alpha i} Z_{\beta j}] = 2^{1-w-d_I/2}.$$
Summing over all indices $(\alpha,i)$, we get 
$$\Delta_I \leq (v d) \cdot 2^{1 - w - d_I/2} = O(w d 2^{-d_I/2}).$$ 

Finally, observe that $\max_{\alpha, i} \E[Z_{\alpha i}] = 2^{-w} \leq 1/2$. Thus, by \tref{th:janson}, 
\begin{equation}\label{eq:janson}
\prod_{\alpha \in I} \E[(1-Z_\alpha)] \leq \E\sbkets{\prod_{\alpha \in I} ( 1 - Z_\alpha)} \leq \exp(2 \Delta_I) \cdot \prod_{\alpha \in I} \E[(1-Z_\alpha)].
\end{equation}

Note that $\E[(1-Z_\alpha)] = (1 - 2^{-w})^v$. Therefore, summing over all sets $I \subseteq [u], |I| = a$ gives us the lower bound of the claim. For the upper bound, we have
\begin{align*}
\E\sbkets{S_a(1-Z_1,\ldots,1-Z_u)} &\leq (1 - 2^{-w})^{va} \cdot \sum_{I \in \binom{[u]}{a}} \exp(2 \Delta_I)\\
&\leq (1 - 2^{-w})^{va} \cdot \sum_{I \in \binom{[u]}{a}} \exp(O(w d 2^{-d_I/2}))) \\
&=  (1 - 2^{-w})^{va} \cdot \binom{u}{a} \cdot \E_{I \in_u \binom{[u]}{a}}\sbkets{  \exp(O(w d 2^{-d_I/2})))}.
\end{align*}
We next bound the last expectation. From the design properties of $\calD$ and a union bound applied to all possible pairs $(\alpha,i), (\beta, j) \in I \times [v]$, it follows that for $I \in_u \binom{[u]}{a}$, 
\begin{align*}
\pr[d_I < k] &= \pr\sbkets{\exists \alpha \neq \beta \in I,\;\; i,j \in [v]\;\; |P^\alpha_i \cap P^\beta_j| \geq w-k}\\
&\leq \sum_{i,j \in [v]} \pr\sbkets{\exists \alpha \neq \beta \in I\;\;|P^\alpha_i \cap P^\beta_j| \geq w-k}\\
&\leq v^2 d^2 \delta.
\end{align*}
Hence, 
\begin{align*}
\E\sbkets{\exp\bkets{O(1) w d  2^{-d_I/2}}} &= \pr[d_I \leq k] \cdot  \E\sbkets{\exp\bkets{O(1) w d  2^{-d_I/2}} \,|\, d_I \leq k} +\\
&\;\;\;\;\;\;\;\;\;\;\; \pr[d_I > k] \cdot  \E\sbkets{\exp\bkets{O(1) w d  2^{-d_I/2}} \,|\, d_I > k}\\
&\leq  \exp\bkets{O(1) w d 2^{-k/2}} + (v^2 d^2 \delta) \cdot \exp\bkets{O(1) w d 2^{-d/2}}\\
&= \exp\bkets{O(1) w d 2^{-k/2}} + 2^{O(w)} \delta.
\end{align*}
This proves the lemma.
\end{proof}
\subsection{Analyzing bias under limited independence}
We next a state a version of \tref{th:biascontrol} for distributions with limited independence. 
\begin{corollary}\label{cor:biascontrol}
Let $\calP = \{P^1,\ldots,P^u\}$ be a collection of partitions of $[n]$ into $w$-sized blocks that is a $(d,k,\delta)$-design and $\epsilon \in (0,1)$. Let $u,v,w$ be such that $v = \Theta(1) w 2^w$ and $1/3 \leq \mathsf{bias}(u,v,w) = (1 - (1- 2^{-w})^v)^u \leq 2/3$. Let $\calD$ be a $(C (w^2 d^2 + wd \log(1/\epsilon)))$-wise independent distribution for $C$ a sufficiently large constant. Then,
$$\abs{\pr_{x \sim \calD}[f_\calP(x) = 1] - \mathsf{bias}(u,v,w)} \leq \epsilon +  \min\begin{cases}O(1) w \exp(-\Omega(d))\\ O(1) \cdot \bkets{w \exp(-\Omega(k)) + \exp(-\Omega(d)) + 2^{w} \delta} \end{cases}. $$
\end{corollary}

We prove the corollary by repeating the proof of \tref{th:biascontrol} while using an analogue of \lref{lm:bias01} for $x \sim \calD$. To do so, we will use the following claim saying that limited independence is sufficient to fool disjunctions of a few Tribes.
\begin{claim}
The following holds for some constant $C > 1$. Let $\{P^1,\ldots,P^r\}$ be a collection of partitions of $[n]$ into $v$ blocks of length $w$ each with $v = \Theta(1) w 2^w$. Then, for any $Cwr(w + r\log(1/\epsilon) + r^2)$-wise independent distribution $\calD$ over $\zo^n$,
$$\abs{\pr_{x \in_u \zo^n}\sbkets{\vee_{\alpha \in [r]} T_{P^\alpha}(x) = 0} - \pr_{x \sim \calD}\sbkets{\vee_{\alpha \in [r]} T_{P^\alpha}(x) = 0}} \leq \epsilon.$$
\end{claim}
\begin{proof}
The lemma follows from some careful approximations of elementary symmetric polynomials. We first setup some notation. Let $k > \max(2r,2e^2 v 2^{-w})$ be an even integer to be chosen later and let $x \sim \calD$ where $\calD$ is a $t$-wise independent distribution for $t \geq kwr$. For $\alpha \in [r]$, let $X^\alpha = (X^\alpha_1,\ldots,X^\alpha_v)$ where $X^\alpha_i = \wedge_{\ell \in P^\alpha_i} x_\ell$. Note that $T_{P^\alpha}(x) = \vee_{i=1}^u X^\alpha_i$. We will use inclusion-exclusion to approximate each $T_{P^\alpha}(x)$.  Let $p_k, e_k:\zo^v \to \zo$ be defined by 
$$p_k(z_1,\ldots,z_v) = \sum_{\ell = 0}^{k-1} (-1)^\ell S_\ell(z_1,\ldots,z_v),\;\;\;\; e_k(z_1,\ldots,z_v) = \sum_{\ell=k}^n (-1)^\ell S_\ell(z_1,\ldots,z_v).$$

\newcommand{\nor}{\neg OR}
By the inclusion-exclusion formula, for all $z \in \zo^v$, $\nor(z) = p_k(z) + e_k(z)$. Further, by Bonferroni inequalities, for all $z \in \zo^v$, 
\begin{equation}\label{eq:kwise1}
0 \leq e_k(z) \leq S_k(z).
\end{equation}
Now,
\begin{align}\label{eq:kwise4} 
1 - \vee_{\alpha \in [r]} T_{P^\alpha}(x) &= \prod_{\alpha \in [r]} \nor(X^\alpha) = \prod_{\alpha \in [r]} \bkets{p_k(X^\alpha) + e_k(X^\alpha)}\\\nonumber
&= \underbrace{\prod_{\alpha \in [r]} p_k(X^\alpha)}_{:= P_k(x)} + \underbrace{\sum_{I \neq \emptyset \subseteq [r]} \prod_{\alpha \notin I} p_k(X^\alpha)\prod_{\alpha \in I} e_k(X^\alpha) }_{:= E_k(x)}.
\end{align}
We view $P_k$ as a low-degree polynomial approximation for the left-hand-side and $E_k$ as the error term. Indeed, $P_k$ is of degree at most $kwr$ in $x$ as each $p_k(X^\alpha)$ is of degree at msot $k$ in $\{X^\alpha_1,\ldots,X^\alpha_v\}$'s which in-turn are of degree at most $w$ in $x$. Therefore, for $t \geq kwr$, $\E[P_k(x)]$ is the same under all $t$-wise independent distributions. We next bound the expectation of $E_k$ under $t$-wise independent distributions. By \eref{eq:kwise1}, $0 \leq e_k(z) \leq S_k(z)$; further, $|p_k(z)| = |\nor(z)- e_k(z)| \leq 1 + e_k(z) \leq 1 + S_k(z)$. Thus, for any fixed $\emptyset \neq I \subseteq [r]$, 
\begin{align*}
\E\sbkets{\abs{\prod_{\alpha \in I} e_k(X^\alpha) \prod_{\alpha \notin I} p_k(X^\alpha)}} &\leq \E\sbkets{\prod_{\alpha \in I} S_k(X^\alpha) \prod_{\alpha \notin I} (1 + S_k(X^\alpha))}\\
&\leq \prod_{\alpha \in I} \E\sbkets{S_k(X^\alpha)^r}^{1/r} \cdot \prod_{\alpha \notin I}\E\sbkets{\bkets{1 + S_k(X^\alpha)}^r}^{1/r}\\
&\text{ (by Fact \ref{fct:Holder})}\\
&\leq \prod_{\alpha \in I} \E\sbkets{S_k(X^\alpha)^r}^{1/r}  \cdot \prod_{\alpha \notin I} \bkets{1 + \E\sbkets{S_k(X^\alpha)^r}^{1/r}}\\
&\text{ (by Minkowski's inequality).}\\
\end{align*}
Note that $S_k(X^\alpha)^r$ is of degree at most $kwr$ as a polynomial in $x$. Therefore, $\E[S_k(X^\alpha)^r]$ is the same under all $t$-wise independent distributions and in particular the same as for $x \in_u \zo^n$. However, in this case $(X^\alpha_1,\ldots,X^\alpha_v)$ are independent indicator random variables with $\E[X^\alpha_i] = 2^{-w}$. Therefore, as $k \geq 2e^2 v 2^{-w}$, by \lref{lm:binomialmoments}, 
$$\E\sbkets{S_k(X^\alpha)^r} \leq (1/2)^{k} \leq 1.$$
Combining the above estimates we get that for all $I \neq \emptyset \subseteq [r]$, 
$$\E\sbkets{\abs{\prod_{\alpha \in I} e_k(X^\alpha) \prod_{\alpha \notin I} p_k(X^\alpha)}} \leq (1/2)^{(|I|)(k/r} \cdot 2^{r-|I|} \leq 2^r \cdot  (1/2)^{k/r}.$$
Plugging the above into equation \eref{eq:kwise4} we get $\E[|E_k(x)|] \leq 2^{2r} (1/2)^{k/r}$. Finally, as the above estimates hold under any $t$-wise independent distribution and $P_k$ is of degree at most $kwr$, we get that 
$$\abs{\pr_{x \in_u \zo^n}\sbkets{\vee_{\alpha \in [r]} T_{P^\alpha}(x) = 0} - \pr_{x \sim \calD}\sbkets{\vee_{\alpha \in [r]} T_{P^\alpha}(x) = 0}} \leq \E_{x \in_u}[|E_k(x)|] + \E_{x \sim \calD}[|E_k(x)|] \leq 2^{2r+1}(1/2)^{k/r}.$$
Plugging in $k = \max(3r(r + \log(1/\epsilon)), 2e v 2^{-w})$ we get the above error to be at most $\epsilon$. It follows that it suffices for $t = kwr$ to be $C(wr(r^2 + r\log(1/\epsilon) + w))$ for a sufficiently big constant $C$.
\end{proof}

\begin{proof}[Proof of \cref{cor:biascontrol}]
The proof of the corollary is similar to that of \lref{lm:bias01} and \tref{th:biascontrol} with one change. Using the notation from the above proof, for any $I \subseteq [u]$, the event $\prod_{\alpha \in I} (1 - Z_\alpha)$ corresponds to the unsatisfiability of $\vee_{\alpha \in I} T_{P^\alpha}$. By using the above theorem, we get an analogue of \eref{eq:janson} for the present case as well: for any $I \subseteq [u]$ with $|I| \leq d$, 
\begin{equation}\label{eq:jansonkwise}
\prod_{\alpha \in I} \E[(1-Z_\alpha)] - \epsilon \leq \E\sbkets{\prod_{\alpha \in I} ( 1 - Z_\alpha)} \leq \exp(2 \Delta_I) \cdot \prod_{\alpha \in I} \E[(1-Z_\alpha)] + \epsilon.
\end{equation}

The corollary now follows by using the above inequality in place of \eref{eq:janson} in the rest of the proof of \tref{th:biascontrol}. 
\end{proof}

\section{Analyzing influence}
\begin{proof}[Proof of \tref{th:influencecontrol}]
Let $Q \subseteq [n]$ with $|Q| = q$. Note that for every $\alpha \in [u]$, a partial assignment $x$ to the variables not in $Q$ leaves $T_{P^\alpha}$ undetermined if and only if
\begin{enumerate}
\item For every part $P^\alpha_j$ that does not intersect $Q$, $x_i = 0$ for some $i \in P^\alpha_j$.
\item For some $j \in [v]$ with $P^\alpha_j \cap Q \neq \emptyset$, $x_i = 1$ for every $i \in (P^\alpha_j \setminus Q)$. 
\end{enumerate}
The above two events are independent of each other. The probability of (1) is at most $(1-2^{-w})^{v-q}$ as there are at least $v-q$ parts of $P^\alpha$ that do not intersect $Q$. The probability of (2) is at most 
\begin{align*}
\sum_{j\in [v]: P^\alpha_j \cap Q \neq \emptyset} 2^{-\bkets{w - |P^\alpha_j \cap Q|}}.
\end{align*}
Therefore, for $\alpha \in_u [u]$, 
\begin{align*}
I_{Q}(f_\calP) &\leq \sum_{\alpha \in [u]} I_{Q}(T_{P^\alpha})\leq u \cdot \E_{u \in_u [u]}\sbkets{ I_{Q}(T_{P^\alpha})}\\
&\leq u \cdot  \E\sbkets{ (1-2^{-w})^{v-q} \cdot \bkets{\sum_{j \in [v]:  P^\alpha_j \cap Q \neq \emptyset} 2^{-w + |P^\alpha_j \cap Q|}}}\\
&\leq u (1-2^{-w})^{v-q} 2^{-w} \cdot \bkets{\sum_{j \in [v]} \E_{\alpha \in_u [u]}\sbkets{\mathsf{1}( P^\alpha_j \cap Q \neq \emptyset) 2^{|P^\alpha_j \cap Q|}}}\\
&\leq u (1-2^{-w})^{v-q} \cdot 2^{-w} \cdot \tau q.\\
\end{align*}
\end{proof}

We next a state a version of \tref{th:influencecontrol} for distributions with limited independence.

\begin{corollary}\label{cor:kwiseinfluence}
Let $\calP = \{P^1,\ldots,P^u\}$ be a collection of partitions of $[n]$ into $w$-sized blocks that is $(q,\tau)$-load balancing. Then, for all $t \geq C w \log(1/\epsilon)$ for some sufficiently big constant $C$, 
$$I_{q,t}(f_\calP) \leq (u (1-2^{-w})^{v-q}) \cdot (\tau 2^{-w}) q + (u v) \epsilon.$$
\end{corollary}
\begin{proof}
The argument is similar to that of \cref{cor:biascontrol}. Let $\calD$ be a $t$-wise independent distribution on $\zo^n$. As in the above proof, observe that $Q$ leaves $T_{P^\alpha}$ undetermined if and only if 
\begin{enumerate}
\item For every part $P^\alpha_j$ that does not intersect $Q$, $x_i = 0$ for some $i \in P^\alpha_j$.
\item For some $j \in [v]$ with $P^\alpha_j \cap Q \neq \emptyset$, $x_i = 1$ for every $i \in (P^\alpha_j \setminus Q)$. 
\end{enumerate}
However, the above two events are no longer independent. For a fixed $\alpha$, let $J \subseteq [v]$ be all parts $P^\alpha_j$ that do not intersect $Q$. Then, the first condition above is equivalent to
$$f_1(x) := \bigwedge_{j \in J} \bkets{\bigvee_{i \in P^\alpha_j} (\neg x_i)}.$$
Similarly, the second condition is equivalent to 
$$f_2(x) := \bigvee_{j \notin J} \bkets{\bigwedge_{i \in P^\alpha_j\setminus Q} x_j}.$$
We are interested in 
$$\pr_{x \lfta \calD}[ f_1(x) \wedge f_2(x) = 1] = \pr_{x \lfta \calD}[f_1(x) = 1] - \pr_{x \lfta \calD}[f_1(x) \wedge (\neg f_2(x)) = 1].$$
Now, observe that
$$f_3(x) := f_1(x) \wedge (\neg f_2(x)) = \bigwedge_{j \in J} \bkets{\bigvee_{i \in P^\alpha_j} (\neg x_i)} \bigwedge_{j \notin J} \bkets{\bigvee_{i \in P^\alpha_j\setminus Q} (\neg x_j)}.$$
Now, $f_1,f_3$ are both read-once CNFs, that is CNF formulas where each variable appears at most once. Therefore, by \tref{th:rcnf}, for $t \gg w \log(1/\epsilon)$, 
$$\abs{\pr_{x \lfta \calD}[f_1(x) = 1] - \pr_{x \in_u \zo^n}[f_1(x) = 1]} \leq \epsilon,\;\;\;\;\abs{\pr_{x \lfta \calD}[f_3(x) = 1] - \pr_{x \in_u \zo^n}[f_3(x) = 1]}  \leq \epsilon. $$
Therefore, 
\begin{multline*}
\pr_{x \lfta \calD}[ f_1(x) \wedge f_2(x) = 1] \leq \pr_{x \in_u \zo^n}[ f_1(x) \wedge f_2(x) = 1] + 2 \epsilon \leq \\
(1-2^{-w})^{v-q} \cdot \bkets{\sum_{j\in [v]: P^\alpha_j \cap Q \neq \emptyset} 2^{-w + |P^\alpha_j \cap Q|} } + 2\epsilon,
\end{multline*}
where the last inequality follows from the arguments of \tref{th:influencecontrol}. The main statement now follows by repeating the calculations of \tref{th:influencecontrol} with the above equation leading to an additional error of $2(u v) \epsilon$.
\end{proof}

\section{Oblivious sampler preserving the moment generating function}\label{sec:expsampler} 
Here we prove \tref{th:expsampler}; it will be the main building block in our final construction of resilient functions. With a view towards future use, we modify the construction presented in the introduction (\eref{eq:expintro}) even if the simpler construction described there suffices for \tref{th:expsampler}. 

Let $\epsilon = \mu/w$, and $c$ a sufficiently big constant to be chosen later. Let $E:[v^c] \times [D] \to [v/D]$ be a $((c \log v)/2, \epsilon)$-strong extractor as in \tref{th:zextractor} with $D = ((\log v)/\epsilon)^C$ for some universal constant $C$. Without loss of generality, suppose that $D$ is prime. For a parameter $\ell \geq 1$ to be chosen later, let $G_\ell:[D]^\ell \to [D]^w$ generate a $\ell$-wise independent distribution as in \lref{lm:kwisecode}. 

Define $G:[v^c] \times [D]^\ell \to [v]^w$ as follows:
\begin{equation}\label{eq:genmain}
G(x,y)_i = G_\ell(y)_i \circ E(x,G_\ell(y)_i),
\end{equation}
where we associate $[D] \times [v/D]$ with $[v]$ in a straightforward manner. 

To analyze the generator we shall use the following lemma about random variables with limited independence. A similar statement appears in \cite{GopalanKM15}; however, our setting is considerably simpler and we give a direct proof in the appendix. 
\begin{lemma}\label{lm:samplerkwise}
Let $Y_1,\ldots,Y_w$ be $\ell$-wise independent random variables supported on $[0,1]$. Then, for all $\theta > 0$, 
$$\E\sbkets{\exp(\theta (Y_1 + \cdots + Y_w))} \leq \prod_i \E[\exp(\theta Y_i)] +\exp(2 \theta w)  \bkets{\frac{e\sum_i \E[Y_i]}{\ell}}^\ell.$$
\end{lemma}

Towards proving \tref{th:expsampler}, we first prove a lemma with some precise but cumbersome bounds on the moment generating function.

\begin{lemma}\label{lm:mtech}
Let $G:[v^c] \times [D]^\ell \to [v]^w$ be as in \eref{eq:genmain} and let $f_1,\ldots,f_w:[v] \to [0,1]$ be functions with $\sum_i \E[f_i] = \mu$. Let $\alpha = G(x,y)$ for $(x,y) \in_u [v^c] \times [D]^\ell$. Then, for all $\theta \geq 0$, 
$$\E_\alpha\sbkets{\exp\bkets{\theta \sum_i f_i(\alpha_i)}} \leq \exp(e^\theta w \epsilon) \cdot \exp(e^\theta \mu) + \exp(2 \theta w) \bkets{\frac{e(\mu+w\epsilon)}{\ell}}^\ell + \exp(\theta w) \cdot (w v^{-c/2}).$$
\end{lemma}
\begin{proof}
Let $Y_i = f_i(\alpha_i)$. Note that for any fixed $x$, the random variables $Y_1,\ldots, Y_w$ are $\ell$-wise independent with respect to the randomness of $y$. We bound the expectation of $\exp\bkets{\theta \sum_i f_i(\alpha_i)} = \exp\bkets{\theta \sum_i Y_i}$ as follows:
\begin{itemize}
\item Using sampling properties of extractors, \lref{lm:extsampler}, w show that with probability at least $1-w v^{-c/2}$ over $x \in_u [v^c]$, $\sum_i \E_y [Y_i] \leq \mu + w \epsilon$. 
\item We then apply the previous lemma to the $Y_i$'s conditioned on $x$ satisfying the above event.
\end{itemize}

For $i \in [w]$, call $x \in [v^c]$ $i$-bad if 
$$\abs{\E[Y_i | x] - \mu_i} = \abs{\E_{z \in_u [D]}[f_i(z \circ E(x,z))] - \mu_i} \geq \epsilon.$$
Call $x \in [v^c]$ \emph{bad} if it is $j$-bad for some $j \in [w]$ and \emph{good} otherwise.
Fix $j \in [w]$. For $z \in [D]$, define $g_z:[v/D] \to [0,1]$ by $g_z(x') = f_j(z \circ x')$. Then, $\sum_{z \in [D]} \E_{x' \in_u [v/D]}[g_z(x')] = D \E_j[f_j] = D\mu_j$; thus, $x$ is $j$-bad if and only if 
$$\abs{(1/D) \sum_{z \in [D]} g_z(E(x,z)) - \mu_j} \geq \epsilon.$$
Therefore, by \lref{lm:extsampler}, for every $j \in [w]$, there are at most $v^{c/2}$ bad strings. Thus, $\pr_{x \in_u [v^c]}[x \text{ is $j$-bad}] \leq v^{-c/2}$. Then, by a union bound, $\pr_{x \in_u [v^c]}[$x$\text{ is bad}] \leq w v^{-c/2}$. Finally, conditioned on $x$ being good, \begin{equation}\label{eq:mtech01}
\sum_j \E_y [Y_j] = \sum_j \E_{z \in_u [D]}[f_j(z \circ E(x,z))] \leq \sum_j (\mu_j + \epsilon) = \mu + w \epsilon.
\end{equation}

Now, conditioned on $x$, the random variables $Y_1,\ldots,Y_w$ are $\ell$-wise independent by the definition of $G(x,y)$. Therefore, by \lref{lm:samplerkwise}, for good $x$,
\begin{align*}
\E_y\sbkets{\exp\bkets{\theta\sum_i f_i(\alpha_i)}} &= \E_y \sbkets{\exp\bkets{\sum_i  Y_i}}\\
&\leq \prod_i \E_y[\exp\bkets{\theta Y_i}] + \exp(2 \theta w) \bkets{\frac{e(\mu+w\epsilon)}{\ell}}^\ell\\
&\leq \prod_i \bkets{1 + (e^\theta)\E_y[Y_i]} + \exp(2 \theta w) \bkets{\frac{e(\mu+w\epsilon)}{\ell}}^\ell\\
&\leq \exp((e^\theta) \sum_i \E_y[Y_i])+ \exp(2 \theta w) \bkets{\frac{e(\mu+w\epsilon)}{\ell}}^\ell\\
&\leq \exp((e^\theta) (\mu + w\epsilon)) + \exp(2 \theta w) \bkets{\frac{e(\mu+w\epsilon)}{\ell}}^\ell.
\end{align*}
Finally, as $\exp(\theta \sum_i f_i(\alpha_i)) \leq \exp(\theta w)$ always, we get that
\begin{align*}
\E_{x,y}\sbkets{\exp\bkets{\theta\sum_i f_i(\alpha_i)}} &\leq \exp((e^\theta) (\mu + w\epsilon))  + \exp(2 \theta w) \bkets{\frac{e(\mu+w\epsilon)}{\ell}}^\ell + \pr[\text{$x$ is bad}] \exp(\theta w)\\
&\leq \exp(e^\theta w \epsilon) \cdot \exp(e^\theta \mu) + \exp(2 \theta w) \bkets{\frac{e(\mu+w\epsilon)}{\ell}}^\ell + \exp(\theta w) \cdot (w v^{-c/2}).
\end{align*}
\end{proof}

\begin{proof}[Proof of \tref{th:expsampler}]
The theorem follows by applying the above lemma with $\epsilon = \mu/w$, $\theta = \ln 2$, $\ell = 12w/(\log w)$. Note that for this setting, there exist extractors as in \eref{eq:genmain} with $D = ((\log v)/\epsilon)^{O(1)}$ so that 
$$\log D = O(\log \log v + \log w + \log(1/\mu)).$$
With this setup, as $\mu \leq 1$, we get 
\begin{align*}
\E_{z \in_u \zo^r}\sbkets{\exp\bkets{\theta\sum_i f_i(G(z)_i)}} &\leq \exp(2w\epsilon) \cdot \exp(2 \mu) + \exp(2w) \bkets{\frac{2e\mu}{\ell}}^\ell + \exp(w) \cdot w v^{-c/2}\\
&\leq 1 + O(\mu) + \exp(w) \cdot w v^{-c/2},
\end{align*}
as $\epsilon \leq \mu/w$. We now set $c = C \max\bkets{1, (w+\log(1/\mu))/(\log v)}$ for a sufficiently large constant $C$ so that the last term is also $O(\mu)$. 

The seed-length of the generator is 
\begin{equation*}
r = c \log v + \ell(\log D) = O(\log v +w + w((\log \log v) + (\log(1/\mu)))/(\log w)).
\end{equation*}

The theorem now follows.
\end{proof}
\section{Explicit resilient functions}
Here we present our main construction proving \tref{th:main}. Fix $v, w$. For a string $\alpha \in [v]^w$, define an associated partition $P^\alpha$ of $[n] \equiv [v w]$ into $w$-sized blocks as follows: 
\begin{itemize} 
\item Write $\{1,\ldots v w\}$ from left to right in $w$ blocks of length $v$ each. Now, permute the $k$'th block by shifting the integers in that block by adding $\alpha_k$ modulo $v$.
\item The $i$'th part now comprises of the elements in the $i$'th position in each of the $w$ blocks.
\end{itemize}
Formally, for $i \in [v]$ $P^\alpha_i = \{(k-1)v + ((i-\alpha_k) \mod v): k \in [w]\}$. As in \cite{ChattopadhyayZ15}, our final function will be $f_\calP$ for $\calP = \{P^\alpha: \alpha \subseteq \calU\}$ for a suitably chosen set of strings $\calU \subseteq [v]^w$. 

\newcommand{\frs}{f_\mathcal{RS}}
\subsection{Polynomially resilient functions from Reed-Solomon code}\label{app:polyresilient}
For intuition, we first use our arguments to present a simpler variant of the construction of \cite{ChattopadhyayZ15} (e.g., the function below is depth $3$ as opposed to the depth $4$ construction of \cite{ChattopadhyayZ15}) to get a $(n^{1-\delta})$-resilient function from Reed-Solomon codes as alluded to in the introduction.

Let $1\leq w \leq v$, where $v$ is prime. For some parameter $\ell \geq 1$ to be chosen later, let $G_\ell:[v]^\ell \to [v]^w$ be an $\ell$-wise independent generator as in \lref{lm:kwisecode} and let $\mathcal{RS} = \{G_\ell(x): x \in [v]^\ell\}$. Let $f \equiv f_\mathcal{RS} = f_\calP$, where $\calP = \{P^\alpha: \alpha \in \mathcal{RS}\}$. We show that for any constant $0 < \beta < 1$, and $\ell \geq 1/2\beta$, $\frs$ is $\Omega(n^{1-\beta})$-resilient and has bias $1/2 \pm n^{-\Omega(1)}$. 

\begin{lemma}\label{lm:frsdesign}
For $1 \leq \ell \leq w$, $\frs$ as defined above is a $(w-\ell)$-design. 
\end{lemma}
\begin{proof}
First note that for any $\alpha, \beta \in [v]^w$ and $i,j \in [v]$,, $|P^\alpha_i \cap P^\beta_j|  = |\{k \in [w]: \beta_k - \alpha_k = (j-i) \mod v\}| \leq w - d_H(\alpha,\beta)$. From the properties of $\mathcal{RS}$ as in \lref{lm:kwisecode}, for $\alpha \neq \beta \in \mathcal{RS}$, $d_H(\alpha,\beta) \geq w-\ell$. Therefore, for $\alpha \neq \beta \in \mathcal{RS}$, $|P^\alpha_i \cap P^\beta_j| \leq w - d_H(\alpha,\beta) \leq \ell$. The claim now follows from the definition of design.
\end{proof}

\begin{lemma}\label{lm:frsbalance}
For $1 \leq \ell \leq w/2$, $\mathcal{RS}$ is $(q,\tau)$-load balancing for $\tau = 2^{\ell} + 2^w (q/v)^{\ell-1}$. 
\end{lemma}
\begin{proof}
The argument here is similar to the proofs of Chernoff bounds for random variables with limited independence (especially those typically used in analyzing limited independence hash functions). 

Let $Q \subseteq [n]$ with $|Q| = q \leq v$ and fix an index $j \in [v]$. Let $\alpha \in_u \mathcal{RS}$ and let $X = |Q \cap P^\alpha_j|$; note that $\E[X] = q/v$. We are interested in estimating $\E[\mathsf{1}(X > 0) 2^X]$. We do so by first proving a tail bound on $X$. To this end,
for $1 \leq i \leq w$, let
$$X_i = \begin{cases} 1 &\text{ if } j-\alpha_i \in Q \cap \{(i-1)v+1,(i-1)v + 2,\ldots,(i v)\}\\
0&\text{ otherwise}\end{cases}.$$
Then, from the definition of the parition $P^\alpha_j$, $X = X_1 + X_2 + \cdots + X_w$. Further, as each $X_i$ only depends on $\alpha_i$, $X_1,\ldots,X_w$ are $\ell$-wise independent. Therefore, by a standard calculation, 
\begin{align*}
\pr[X \geq \ell]  &\leq \E[S_\ell(X_1,\ldots,X_w)]= \sum_{I \subseteq [w], |I| = \ell} \E\sbkets{\prod_{i \in I} X_i}\\
&= \sum_{I \subseteq [w], |I| = \ell} \prod_{i \in I} \E[X_i] \leq \binom{w}{\ell}\bkets{\frac{\sum_{i=1}^w\E[X_i]}{w}}^\ell.
\end{align*}
where the last inequality follows from Fact \ref{fact:symmax}. Now, as $\E[X] = q/v$, the above expression simplifies to 
$$\pr[X> \ell] \leq \binom{w}{\ell} \cdot (q/wv)^\ell \leq \bkets{\frac{eq}{v\ell}}^\ell.$$

Thus, for $\ell \geq 3$,
\begin{align*}
\E\sbkets{1(Q \cap P^\alpha_j \neq \emptyset) 2^{|Q \cap P^\alpha_j|}} &= \E\sbkets{1(X > 0) 2^X}  \\
&= \pr[ X < \ell] \E\sbkets{1(X > 0) 2^X | X < \ell} + \pr[X\geq \ell] \E\sbkets{1(X > 0) 2^X | X \geq \ell}  \\
&\leq 2^\ell \pr[X > 0] + 2^w \pr[X > \ell] \\
&\leq 2^\ell (q/v) + 2^w (q/v)^{\ell}. 
\end{align*}
The claim now follows from the definition of load balancing.
\end{proof}

We next use the above claims along with Theorems \ref{th:biascontrol} and \ref{th:influencecontrol} for a suitable setting of parameters. 
\begin{lemma}\label{lm:frsanalysis} 
For all $0 < \delta < 1$, there exists a constant $c_\delta \geq 1$ and a suitable choice of $v = \Theta_\delta( 2^w w)$ such that the following holds. For $\ell \geq 1/\delta$, the function $\frs$ as defined above is $c_\delta 2^{-w}$-strongly resilient and $\pr_{x \in_u \zo^n}[\frs(x) = 1] = 1/2 \pm 2^{-\Omega(w)}$. 
\end{lemma}
\begin{proof}
Let $\ell = \max(3, \lceil 1/\delta \rceil)$. For $v$ to be chosen in a little bit, let $u = v^\ell$ and $f \equiv \frs$. We would like our choice of $v$ to minimize $| v - 2^w \ln ((\ln 2)/v^\ell)|$ so that we can get an almost-balanced function using Fact \ref{fact:biasapprox}. To this end, let $\phi: \R_+ \to \R$ be defined by 
$$\phi(x) = x -   2^w \ln((\ln 2) x^\ell)$$
and let $x^* \geq 1$ be such that $\phi(x^*) = 0$. There exists such an $x^*$ by the continuity of $\phi$. It is also easy to check that for $w$ sufficiently large, $\phi'(y) \geq 0$ for all $y \geq x^*$. We set $v$ to be the smallest prime larger than $x^*$.\footnote{We can find such a prime in time $2^{O(w)}$ which is fine for us.} Note that $v \leq x^* + B$ where $B = 2^{c_1 w}$ for some universal constant $c_ 1 < 1$ (see \cite{Wikiprimegap} for instance), so that $0 = \phi(x^*) \leq \phi(v) \leq \phi(x^*) + B$. Let $\theta = (1- 2^{-w})^v$. Then, by Fact \ref{fact:biasapprox}, $(1+\theta)^u = O(1)$ and
$$\mathsf{bias}(u,v,w) = (1-\theta)^u = 1/2 \pm O(1) 2^{-\Omega(w)}.$$

Now, by \lref{lm:frsdesign} and \tref{th:biascontrol}, 
$$\pr_{x \lfta \calD_{1/2}}[\frs(x) = 1] = \mathsf{bias}(u,v,w) \pm \exp(-\Omega(w-\ell)) = \mathsf{bias}(u,v,w) \pm \exp(-\Omega(w)).$$

Next, by \lref{lm:frsbalance} and \tref{th:influencecontrol}, for any $q \leq 2^{w(1-\delta)}$, as $\ell \geq \lceil 1/\delta \rceil$, 
\begin{align*}
I_q(\frs) &\leq u ( 1- 2^{-w})^{v-q} \cdot (2^{-w}q) \cdot \bkets{2^{\ell} + 2^w (q/v)^{\ell-1}}\\
&\leq O(1) \cdot (2^{-w}q) \cdot (2^\ell + 2^w 2^{-\ell \delta w}) \\
&= O_\delta(2^{-w}q).
\end{align*}

Therefore, $\frs$ is $O_\delta(2^{-w})$-strongly resilient.
\end{proof}

\begin{corollary}\label{cor:polyresilientkwise}
For all $0 < \delta < 1$, there exists a constant $c_\delta \geq 1$ such that the following holds. There exists an explicit depth-three monotone function $f:\zo^n \to \zo$ which can be computed in time $n^{c_\delta}$ such that for $t \geq c_\delta (\log n)^{4}$
\begin{itemize}
\item $f$ is almost balanced: for any $t$-wise independent distribution $\calD$ on $\zo^n$, $\pr_{x \lfta \calD}[f(x) =1 ] = 1/2 \pm n^{-\Omega(1)}$.
\item $f$ is $t$-wise $c_\delta n^{-(1-\delta)}$-strongly resilient. 
\end{itemize}
\end{corollary}
\begin{proof}
We instantiate the previous lemma to get $f \equiv \frs$ for $\delta' = \delta/2$. Then, $n = O(2^w w^2)$ so that $2^{w(1-\delta')} = \Omega( n^{1-\delta})$. To analyze the bias under $t$-wise independent distributions, we apply \cref{cor:biascontrol} with $d = w/2$, $\epsilon = 1/n$ instead of \tref{th:biascontrol} in the above argument. Similarly, to analyze the influence under $t$-wise independent distributions we use \cref{cor:kwiseinfluence} with $\epsilon' = 1/(n^3)$ instead of \tref{th:influencecontrol}. Then, the amount of independence needed is $ O(w^2 d^2 + wd \log(1/\epsilon)) = O(\log^4 n)$. 
\end{proof}

\subsection{Proof of \tref{th:main}}
We now prove \tref{th:main}. The approach is similar to the above with one crucial difference: we use the output of the generator from \tref{th:expsampler} instead of the Reed-Solomon code. We ensure the requisite design properties to apply \tref{th:biascontrol}, at a high-level, by \emph{padding} the output of the generator with a Reed-Solomon code.

 Let $c$ be a sufficiently big constant to be chosen later and suppose that $v,D$ are prime numbers below. Let $E:[v^c] \times [D] \to [v/D]$ be a strong extractor with error $\epsilon = \poly(1/w)$ and $D = \poly(w)$ to be chosen later. Let $\ell = \Theta(w/(\log w))$ be a parameter to be chosen later. Let $G_c:[v]^c \to [v]^{2c}$, $G_\ell:[D]^\ell \to [D]^{w-2c}$ generate a $c$-wise independent distribution over $[v]$ and a $\ell$-wise independent distribution over $[D]$ respectively as guaranteed by \lref{lm:kwisecode}. 

Now, define $\calU:[v^c] \times [D]^\ell \to [v]^w$ as follows: 
\begin{equation}\label{eq:mainUdef}
\calU(x,y)_i = \begin{cases}
G_c(x)_i & \text{ if $1 \leq i \leq 2c$}\\
G_\ell(y)_{i-2c} \circ E(x,G_\ell(y)_{i-2c}) & \text{ if $2c < i \leq w$}
\end{cases}.
\end{equation}
(Here, we associate an element of $[D] \times [v/D]$ with an element of $v$ in a straightforward bijective manner.) 

Abusing notation, we let $\calU = \{\calU(x,y): x \in [v^c], y \in [D]^\ell\} \subseteq [v]^w$ as well. Our final function will be $f \equiv f_\calP$ for $\calP = \{P^\alpha: \alpha \in \calU\}$. The following claims help us apply \tref{th:biascontrol} and \tref{th:influencecontrol} to analyze $f_\calP$. 

\begin{lemma}\label{lm:udesign}
For all $c < (w- \ell)/3$, $\calP$ is a $(c, w-2c-\ell, 1/D^\ell)$-design. 
\end{lemma}
\begin{proof}
Note that for any $\alpha, \beta \in [v]^w$, and any $i, j \in [v]$, $|P^\alpha_i \cap P^\beta_j| = |\{k \in [w]: \beta_k - \alpha_k = (j-i) \mod v\}| \leq w - d_H(\alpha,\beta)$. 

Let $\alpha = G(x,y) \neq \beta = G(x',y') \in \calU$. We consider two cases depending on whether $y \neq y'$. If $y \neq y'$, then $d_H(\alpha, \beta) \geq d_H(G_\ell(y), G_\ell(y')) \geq w-2c - \ell$; similarly, if $x \neq x'$, then $d_H(\alpha, \beta) \geq d_H(G_c(x), G_c(x')) \geq c$. Therefore, $d_H(\alpha, \beta) \geq \min(c, w-2c-\ell) = c$. Hence, $\calP$ is a $c$-design. 

Now, consider any fixed $\alpha = G(x,y) \in \calU$, $i, j \in [v]$, and $\beta = G(x',y') \in_u \calU$. Then, by the above argument, if $y \neq y'$, $d_H(\alpha, \beta) \geq w-2c-\ell$ so that $|P^\alpha_i \cap P^\beta_j| \leq 2c + \ell$. On the other hand, $\pr[y'\neq y] \geq 1- 1/D^\ell$. Thus, 
$$\pr_{\beta \in_u \calU}[|P^\alpha_i \cap P^\beta_j| \geq w - (w-2c-\ell)] \leq \pr[ y' = y] \leq 1/D^\ell.$$
Therefore, $\calP$ is a $(c,w-2c-\ell,1/D^\ell)$-design as needed. 
\end{proof}

\ignore{
\begin{lemma}\label{lm:wcdistance}
For $(x,y) \neq (x',y') \in [v]^c \times [D]^\ell$, $d_H(\calU(x,y),\calU(x',y')) \geq \min(c,w-2c-\ell)$. 
\end{lemma}
\begin{proof}
If $y \neq y'$, then $d_H(\calU(x,y), \calU(x',y')) \geq d_H(G_\ell(y), G_\ell(y')) \geq w-2c-\ell$; similarly, if $x \neq x'$, then $d_H(\calU(x,y), \calU(x',y')) \geq d_H(G_c(x), G_c(x')) \geq c$. 
\end{proof}

\begin{lemma}\label{lm:acdistance}
If $c < w - 2c - \ell$, then for every $(x,y) \in [v]^c \times [D]^\ell$, $(x',y') \in_u [v]^c \times [D]^\ell$, $\pr_{(x',y')}[d_H(\calU(x,y), \calU(x',y')) < w-2c-\ell] \leq 1/D^\ell$. 
\end{lemma}
\begin{proof}
From the above argument, the required condition holds unless $y \neq y'$; on the other hand, $\pr[y' = y] \leq 1/D^\ell$. 
\end{proof}}

We next use \lref{lm:mtech} to analyze the load-balancing properties of $\calP$.
\begin{lemma}\label{lm:ubalance}
Let $\calU$ be as in \eref{eq:mainUdef} for $E$ being a $((c \log v)/2, \epsilon)$-extractor for $\epsilon < 1/w$ and $\ell \geq 6w/(\log w)$. Then, $\calP = \{P^\alpha: \alpha \in \calU\}$ is $(q,\tau)$-load balancing for $q \leq v$ and $\tau = O(2^{2c} + (vw 2^{w}) v^{-c/2} )$. 
\end{lemma}
\begin{proof}
Let $Q \subseteq [n]$ with $|Q| = q \leq v$ and fix an index $j \in [v]$. Let $\alpha \in_u \mathcal{U}$ and let $X = |Q \cap P^\alpha_j|$; note that $\E[X] = q/v$. For $1 \leq i \leq w$, let
$$X_i = \begin{cases} 1 &\text{ if } j-\alpha_i \in Q \cap \{(i-1)v+1,(i-1)v + 2,\ldots,(i v)\}\\
0&\text{ otherwise}\end{cases}.$$
Then, from the definition of the parition $P^\alpha_j$, $X = X_1 + X_2 + \cdots + X_w$.

For brevity, let $\E[X] = \mu = q/v \leq 1$, $Y = \sum_{i=1}^{2c} X_i$, and $Z = \sum_{i=2c+1}^{w} X_i$. Note that $X = Y+Z$.  We are interested in estimating $\E[\mathsf{1}(X > 0) 2^X]$. We do so mainly by applying \lref{lm:mtech} to the random variable $Z$ combined with the trivial observation that $Y$ is at most $2c$: 
\begin{align*}
\E_{\alpha \in_u U}\sbkets{1(Q \cap P^\alpha_k \neq \emptyset) 2^{|Q \cap P^\alpha_k|}} &= \E\sbkets{1(X > 0) 2^X} = \E\sbkets{2^X} - \pr[X = 0]\\
&= \E\sbkets{2^X} - 1 + \pr[X \geq 1] \\
&\leq \E\sbkets{2^X} - 1 + \E[X]\\
&= \E[2^Y-1] + \E[2^Y (2^Z-1)] + \mu\\
&\leq \E[2^Y-1] + 2^{2c}\cdot \E[2^Z-1] + \mu \\
&\leq \E[2^{2c}Y] + 2^{2c} \cdot \E[2^Z-1] + \mu. 
\end{align*}

By \lref{lm:mtech} applied to $Z$ with $\theta = \ln 2$, $\ell = 12w/(\log w)$, and $\epsilon < 1/2w$, we get that 

\begin{align*}
\E[2^Z] &= O(1) \exp(O(\mu)) + 2^w \bkets{\frac{2e}{\ell}}^\ell+ O(w 2^w) v^{-c/2}  \\
&\leq 1 + O(\mu) + 2^{-w} + O(w 2^v) v^{-c/2}.
\end{align*}

Further, $\E[Y] \leq \E[X] = \mu$. Therefore, as $\mu \geq 1/v$, 
$$\E\sbkets{1(Q \cap P^\alpha_k \neq \emptyset) 2^{|Q \cap P^\alpha_k|}}  \leq O(2^{2c}) \mu + O(1) w 2^{w} v^{-c/2} \leq O(1) \cdot \mu \cdot (2^{2c} + w 2^w v^{-c/2 + 1}).$$
\end{proof}

\begin{proof}[Proof of \tref{th:main}]
The theorem follows essentially by combining the above two lemmas and \tref{th:biascontrol}, \tref{th:influencecontrol}. We first set up some parameters. Let $c \geq 2$ be a sufficiently large constant to be chosen later. Let $w \geq 1$ be arbitrary and $\epsilon = 1/w^3$. Let $D = ((c\log v)w)^C$ for some universal constant $C$ so that there exists an explicit $(c (\log v)/2,\epsilon)$-strong extractor $E:[v^c] \times [D] \to [v/D]$ for all $c \geq C$ as in \tref{th:zextractor}. Set $v = \Theta(2^w w)$ to be chosen precisely in a little bit. For this setting of $v$, let $\calU \subseteq [v]^w$ be as defined in \eref{eq:mainUdef} with $\ell = 12w/(\log w)$ and $E$ as the extractor. Let $\calP = \calP_\calU$. Then, $|\calU| := u = v^c \times D^\ell$. We will show that $f \equiv f_\calP$ satisfies the conditions of \tref{th:main} for $c$ sufficiently large.

As in the proof of \lref{lm:frsanalysis}, we would like $v$ to be as close as possible to $2^w \ln ((\ln 2) u)$. To this end, let $\phi: \R_+ \to \R$ be defined by 
$$\phi(x) = x -   2^w \bkets{c \ln x + C \ell \ln (\log x) + C \ell \ln (cw) + \ln \ln 2}$$
and let $x^* \geq 1$ be such that $\phi(x^*) = 0$. There exists such an $x^*$ by the continuity of $\phi$. Let $v$ be the smallest prime larger than $x^*$. Note that $v \leq x^* + B$ where $B = 2^{c_1 w}$ for some universal constant $c_ 1 < 1$ (see \cite{Wikiprimegap} for instance), so that $0 = \phi(x^*) \leq \phi(v) \leq \phi(x^*) + B$. Let $\theta = (1- 2^{-w})^v$. Then, by Fact \ref{fact:biasapprox}, $(1+\theta)^u = O(1)$ and
\begin{equation}\label{eq:mainbiasnumerical} 
\mathsf{bias}(u,v,w) = (1-\theta)^u = 1/2 \pm O(1) 2^{-\Omega(w)}.
\end{equation}

\paragraph{Analyzing bias:} By \lref{lm:udesign}, $\calP$ is a $(c,w-2c - \ell,1/D^\ell)$-deisgn. Therefore, by \tref{th:biascontrol},

\begin{align*}
\pr_{x \in_u \zo^n}[f_\calP(x) = 1] &= \mathsf{bias}(u,v,w) + w \exp(-\Omega(w)) + \exp(-\Omega(c)) + 2^w/D^\ell\\
&= \mathsf{bias}(u,v,w) + \exp(-\Omega(c)),
\end{align*}
as $D^\ell \gg 2^w$ for $C$ a sufficiently large constant. 

\paragraph{Analyzing influence:} We claim that $f_\calP$ has small influence for coalitions of size $o(2^w)$. Let $1 \leq q \leq v$ so that $q/v \leq 1$. Then, by \lref{lm:ubalance}, $\calP$ is $(q,\tau)$-load balancing for 
$$\tau = O(2^{2c}) (1 + v^{-c/2+1}) = O(2^{2c}).$$
Therefore, by \tref{th:influencecontrol}, for all $q \leq v$, 
$$I_{q}(f_\calP) \leq u (1-2^{-w})^{v-q} \cdot 2^{2c} (2^{-w} q) = O(2^{2c}) 2^{-w} q = O(2^{2c}) \cdot \bkets{ (\log^2 n)/n}\cdot q.$$
Here, the last inequality follows as $2^w = \Theta(n/(\log^2 n))$. 

The theorem now follows by choosing $c$ to be a sufficiently large constant.
\end{proof}

We next prove \tref{th:mainkwise}.
\begin{proof}[Proof of \tref{th:mainkwise}]
The proof is exactly the same as the above argument for \tref{th:main} but instead of using \tref{th:biascontrol} we use \cref{cor:biascontrol} with $d = c$, $\epsilon = 1/u$ and instead of \tref{th:influencecontrol} we use \cref{cor:kwiseinfluence} with $\epsilon = 1/(uv^2)$. The amount of independence we need is $t \gg O(w^2 d^2 + wd \log(1/\epsilon)) = O(\log^2 n)$ as required for the theorem. Note that the \emph{error} $1/9$ can be made to be an arbitrary small constant. 
\end{proof}

\ignore{
For this, we need the following theorem of \cite{Braverman10, Tal14a}.
\begin{theorem}
There exists a constant $C$ such that for all depth $d$ AC0 functions $f:\zo^n \to \zo$ of size $m$ and $t \geq C (\log(m/\epsilon))^{3d+3}$, for every $t$-wise independent distribution $\calD$ on $\zo^n$, 
$$\abs{\pr_{x \lfta \calD_{1/2}}[f(x) = 1] - \pr_{x \lfta \calD}[f(x) = 1]} \leq \epsilon.$$
\end{theorem}

Applying the above theorem to the depth three function $f$ given by \tref{th:main} immediately gives \cref{cor:resilientkwise}.}

\section{Better two-source extractors}
Our improved quantitative bounds for two-source extractors follow immediately by using our resilient functions, \tref{th:mainkwise} and \cref{cor:polyresilientkwise} in the reduction of \cite{ChattopadhyayZ15}. Concretely, \cite{ChattopadhyayZ15} show the following for some universal constants $0 < c < 1$ and $C,C' \geq 1$. Suppose for some functions $\epsilon:\mathbb{Z}_+ \to [0,1]$, and $t:\mathbb{Z}_+ \to \mathbb{Z}_+$, the following holds: for all $m \geq 1$, there exists an explicit $t(m)$-wise $(m^{1-c}, \epsilon)$-resilient function $f:\zo^m \to \zo$ with $\pr_{x \in_u \zo^m}[f(x) = 1] = 1/2 \pm \epsilon(m)$. Then, there exists an explicit $(n,k)$ two-source extractor with error at most $\epsilon(n^C)$ and $k \geq C' \cdot (t(n^{C}))^4 \log^2 n$. Instantiating this reduction with the resilient function from \tref{th:mainkwise} (with $1/9$ replaced with a sufficiently small constant) gives \tref{th:ext1}; similarly, using \cref{cor:polyresilientkwise} gives \tref{th:ext2}. \bibliographystyle{alpha}
\bibliography{refpedia}
\appendix
\section{Missing proofs}
\begin{proof}[Proof of \lref{lm:binomialmoments}]
For a set $T \subseteq [v]$, let $n(T) = |\{(I_1,\ldots,I_r) : |I_j| = k, \forall j \in [r] \text{ and } \cup_{j=1}^r I_j = T\}|$. Then, 
$$S_k(X_1,\ldots,X_v)^r = \sum_{I_1,\ldots,I_r \in \binom{[v]}{k}} \prod_{i \in \cup_{j=1}^r I_j} X_i = \sum_{T \subseteq [v]} n(T) \prod_{i \in T} X_i.$$
Thus, 
$$\E[S_k(X_1,\ldots,X_v)^r] = \sum_{T \subseteq [v]} n(T) p^{|T|}.$$
Finally, observe that for any $T$ with $|T| = t$, $n(T) \leq \binom{t}{k}^r \leq (et/k)^r$. Thus,
\begin{align*}
\E[S_k(X_1,\ldots,X_v)^r] &\leq \sum_{t=k}^{kr} \binom{v}{t} \cdot p^t \bkets{\frac{et}{k}}^r\\
&\leq \sum_{t=k}^{kr}\bkets{\frac{evp}{t}}^t \cdot \bkets{\frac{et}{k}}^r \leq \sum_{t=k}^{kr} \bkets{\frac{1}{2}}^r \cdot \bkets{\frac{evp}{t}}^{t-r} \leq (1/2)^{k},
\end{align*}
as $e vp /t \leq 1/4$, and $e^2 vp/k \leq 1/2$. 
\end{proof}

\begin{proof}[Proof of Fact \ref{fact:biasapprox}]
Note that under the assumptions we must have $v = \Omega(2^w w)$. Then, by Fact \ref{fact:expapprox}, 
$$(1-2^{-w})^{B + 2^w \ln(u/(\ln 2))} \leq \theta \leq (1-2^{-w})^{2^w \ln(u/(\ln 2))} \leq (\ln 2)/u.$$

Therefore, $(1 + \theta)^u = O(1)$ and 
$$(1-\theta)^u \geq (1- (\ln 2)/u)^u \geq 1/2 (1 - (\ln 2)/u)^{\ln 2} \geq 1/2 (1 - O(1)/u).$$
Further,
$$\theta \geq ((\ln 2)/u) \cdot (1-2^{-w})^{B + \ln((\ln 2) u)} \geq ((\ln 2)/u) \cdot ( 1- O(B + \ln u) 2^{-w}).$$ 
Thus, 
$$(1-\theta)^u \leq \exp(-u \theta) \leq (1/2)^{\bkets{1 - O(B\ln u) 2^{-w}}} \leq 1/2 \pm O(B \ln u) 2^{-w}.$$
\end{proof}

\begin{proof}[Proof of \lref{lm:samplerkwise}]
The proof relies on the following elementary inequality about symmetric polynomials: for all $0 \leq a_1,\ldots,a_w \leq A$,
$$\prod_i (1 + a_i) \leq \sum_{i=1}^{\ell-1} S_i (a_1,\ldots,a_w) + (1 + A)^w \cdot S_\ell(a_1,\ldots,a_w).$$

Let $Z_i = \exp(\theta Y_i) - 1$. Let $p(Z) = \sum_{i=1}^{\ell-1} S_i (Z_1,\ldots,Z_w)$. Applying the above inequality to $Z = (Z_1,\ldots,Z_w)$ we get
\begin{multline*}
\E\sbkets{\exp(\theta (Y_1 + \cdots + Y_w))} = \E\sbkets{\prod_i (1 + Z_i)} \leq \E[p(Z)] + \exp(\theta w) \cdot \E[S_\ell(Z)] \leq \\\E[p(Z)] + \exp(\theta w) \cdot \binom{w}{\ell} \cdot \bkets{\frac{\sum_i \E[Z_i]}{w}}^\ell,
\end{multline*}
where the last inequality follows because $Z$'s are $\ell$-wise independent so that by Fact \ref{fact:symmax}
$$\E[S_\ell(Z)] = \sum_{I \subseteq [w], |I| = \ell} \prod_{i \in I} \E[Z_i] \leq \binom{w}{\ell} \cdot \bkets{\frac{\sum_i \E[Z_i]}{w}}^\ell.$$
Now, for any $i \in [w]$,
$$\E[Z_i] = \sum_{j=1}^\infty \frac{\theta^j \E[Y_i^j]}{j!} \leq \sum_{j=1}^\infty \frac{\theta^j \E[Y_i]}{j!} = \E[Y_i] (e^\theta - 1).$$
Therefore, $\sum_i \E[Z_i] \leq e^\theta (\sum_i \E[Y_i]) = e^\theta \cdot \mu$. Finally, note that as $\{Z_1,\ldots,Z_w\}$ are $\ell$-wise independent, the expectation of $p(Z)$ would be the same as when the $Z$'s were independent. However, in this case, 
$$\E[p(Z)] \leq \E\sbkets{\prod_i (1+ Z_i)} = \prod_i \E[(1 + Z_i)].$$
Combining the above equations, we get
\begin{align*}
\E\sbkets{\exp(\theta (Y_1 + \cdots + Y_w))} &\leq \prod_i \E[\exp(\theta Y_i)] + \exp(\theta w) \cdot \binom{w}{\ell} \cdot (e^\theta\mu/w)^\ell\\
&\leq \prod_i \E[\exp(\theta Y_i)] + \frac{\exp(2\theta w) (e \mu)^\ell}{\ell^\ell}.
\end{align*}
This finishes the claim.

\end{proof}

\end{document}